\documentclass[twoside,11pt]{article}

\usepackage{jmlr2e}
\usepackage{amsmath, amsfonts}
\usepackage{caption}
\usepackage{subcaption}

\hypersetup{
    colorlinks=true,
    allcolors=blue,
    pdfpagemode=FullScreen,
}

\def\mc{\mathcal}
\def\R{\mathbb{R}}

\ShortHeadings{Filament Plots for Data Visualization}{Strawn}
\firstpageno{1}

\begin{document}

\title{Filament Plots for Data Visualization}

\author{\name Nate Strawn \email nate.strawn@georgetown.edu \\
       \addr Department of Mathematics and Statistics\\
       Georgetown University\\
       Washington, D.C. 20037, USA}

\editor{}

\maketitle

\begin{abstract}%   <- trailing '%' for backward compatibility of .sty file
The efficiency of modern computer graphics allows us to explore collections of space curves simultaneously with ``drag-to-rotate" interfaces. This inspires us to replace ``scatterplots of points" with ``scatterplots of curves" to simultaneously visualize relationships across an entire dataset. Since spaces of curves are infinite dimensional, scatterplots of curves avoid the ``lossy" nature of scatterplots of points. In particular, if two points are close in a scatterplot of points derived from high-dimensional data, it does not generally follow that the two associated data points are close in the data space. Andrews plots provide scatterplots of curves that perfectly preserve Euclidean distances, but simultaneous visualization of these graphs over an entire dataset produces ``visual clutter" because graphs of functions generally overlap in 2D. We mitigate this ``visual clutter" issue by constructing computationally inexpensive 3D extensions of Andrews plots. First, we construct optimally smooth 3D Andrews plots by considering linear isometries from Euclidean data spaces to spaces of planar parametric curves. We rigorously parametrize the linear isometries that produce (on average) optimally smooth curves over a given dataset. This parameterization of optimal isometries reveals many degrees of freedom, and (using recent results on generalized Gauss sums) we identify a particular member of this set which admits an asymptotic ``tour" property that avoids certain local degeneracies as well. Finally, we construct unit-length 3D curves (filaments) from Bishop frames induced by 3D Andrews plots. We conclude with examples of filament plots for several standard datasets\footnote{Code at \href{https://github.com/n8epi/filaments}{https://github.com/n8epi/filaments} and examples at \href{https://n8epi.github.io/filaments/}{https://n8epi.github.io/filaments/}}, illustrating how filament plots avoid ``visual clutter".\\
\end{abstract}


\section{Introduction}

A traditional tool in the data scientist's toolbox is the use of a ``scatterplot". It allows one to identify relationships (if any) between two quantitative variables, and make a statistical determination of the model that should be used in analyzing the data. Advances in computing capabilities allow for much more. Instead of illustrating approximate or incomplete similarities and differences of high-dimensional datasets using scatterplots of points in the plane, we suggest illustrating exact similarities and differences using ``scatterplots of curves in space" or ``filament plots". Modern computer graphics allows us to interactively view simultaneous plots of substantial collections of space curves using a ``drag-to-rotate" interface, and this paper introduces a way to compute aesthetically coherent space curves from a dataset using little more than the SVD of a data matrix. Of course, this may be coupled with fast Johnson-Lindenstrauss transforms \cite{ailon2009fast, jin2019faster} when the dimension of a dataset is prohibitively large. Figure \ref{fig:irisAndrews2D} depicts the graphs of our curves for the Iris dataset \cite{anderson1936species} (a 4D dataset of metric measurements of different iris flowers). 
\begin{figure}[ht]
\centering
\includegraphics[scale=0.6]{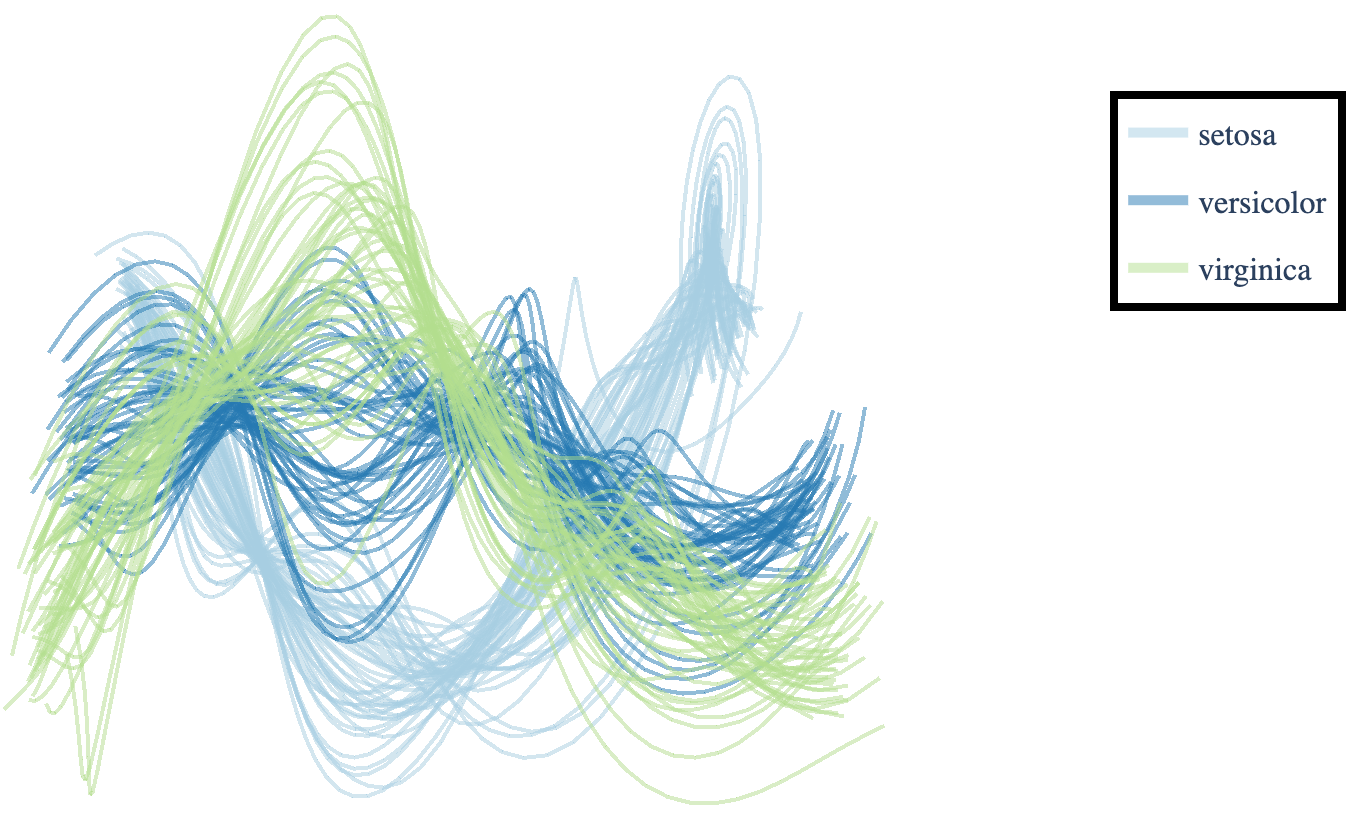}
\caption{The graphs of the 3D Andrews plots for the Iris dataset.}
\label{fig:irisAndrews2D}
\end{figure}

Because the plots of Figure \ref{fig:irisAndrews2D} are graphs of 2D curves with a third axis for the time parameter, tangent spaces of the curves always partially align. This in combination the fact that the 2D curves are trigonometric polynomials tends to produce plots that ``braid" together and partially occlude one another at any viewing angle. If they evolved in ``bushy" manner, then the ``visual clutter" induced by partial occlusions might be reduced. For the additional cost of computing numerical solutions to Bishop frame equations \cite{bishop1975there} for each of these curves, we perform a non-linear transformation of the Andrews curves to obtain the filament plots of Figure  \ref{fig:irisFilament}. In particular, these filament plots are ``bushy", and hence exhibit less ``visual clutter" along many viewing angles.  Since the images in this paper do not support ``drag-to-rotate" interfaces, we encourage interested readers to view the examples on Github\footnote{\href{https://n8epi.github.io/filaments/}{https://n8epi.github.io/filaments/}}. 

\begin{figure}[ht]
\centering
\includegraphics[scale=0.6]{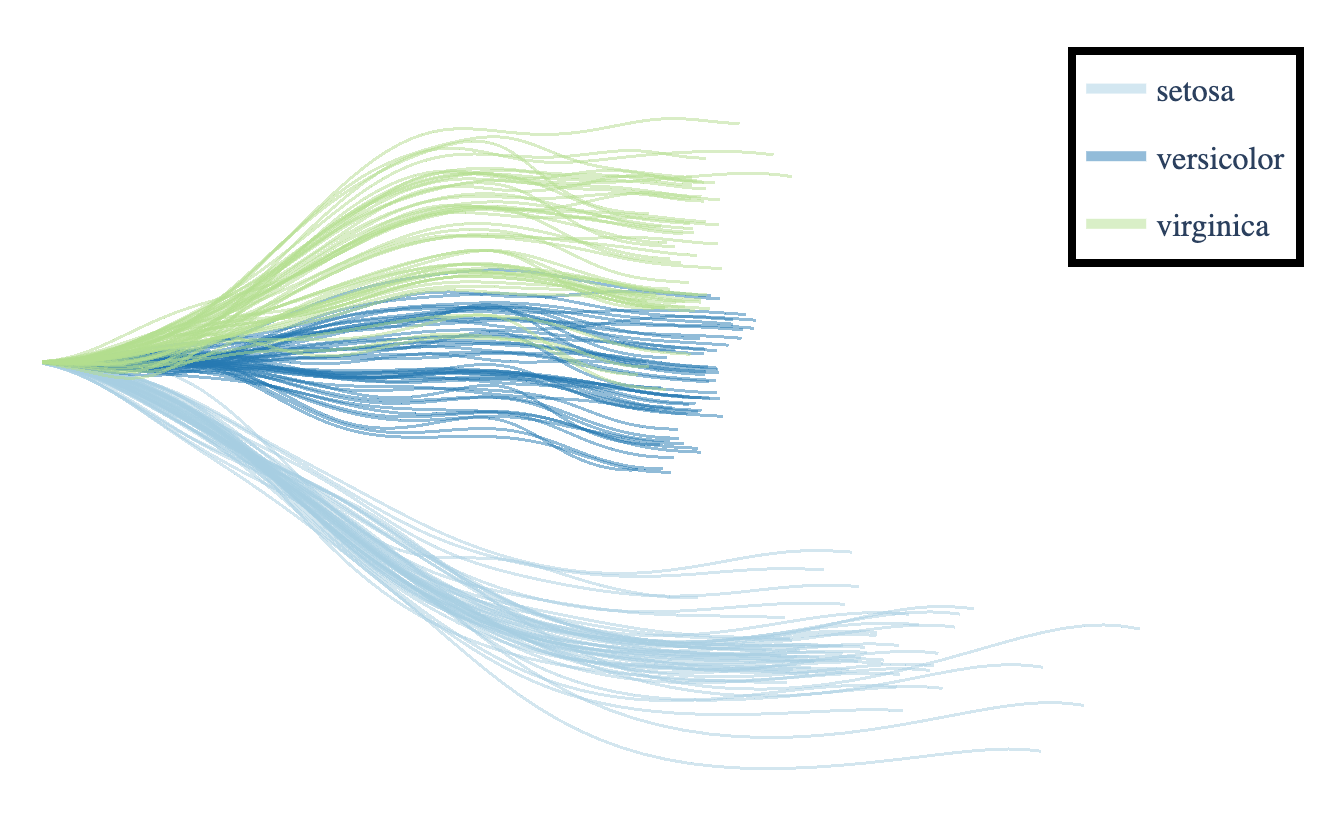}
\caption{The Filament plot for the Iris dataset obtained by solving Bishop frame equations using the 3D Andrews plots as ``symmetric curvature functions". Note that this is just a single perspective, so interested readers should view the code to see the plot in a ``drag-to-rotate" interface.}
\label{fig:irisFilament}
\end{figure}

\paragraph{Defining ``visual clutter".} We used the term ``visual clutter" several times already and we continue to invoke ``visual clutter" to justify many choices throughout this paper. A quantitative approach to ``visual clutter" remains outside the scope of this paper, but a qualitative description of ``visual clutter" is obtained by considering two curves from different data points. If the curves are well separated for their entire length, then we may say that the full image of these curves is not cluttered. If the curves overlap or follow each other closely for much of their length, then we consider the curves to be cluttered. In particular, if an observer stands at a distance from the visualization of the two curves, the observer may not distinguish between the curves as they separate. 

For example, the setosa and virginica examples produce well-separate curves in Figure \ref{fig:irisFilament}, so it is easy to visually identify the sets of curves from the two classes. However, there are some versicolor and virginica examples that require more attention to separate. Qualitatively, we say that these examples are cluttering each other. Overall, we say that the image exhibits low visual clutter.

On the other hand, many of the individual curves in the 3D Andrews plots of Figure \ref{fig:irisAndrews2D} cross many of the other curves in the visualization. Because of this, the separation of the setosa and virginica examples is not immediately clear. Also, most of the verisicolor examples cross many of the virginica examples. Overall, we say that this image exhibits a fair amount of visual clutter.

\paragraph{Contributions.} 

In this work, we first identify several conditions on linear operators from Euclidean data spaces to spaces of plane curves that encourage low ``visual clutter":
\begin{enumerate}
\item Isometry: the Euclidean distances between any two data points should match the ``distance" between the curves that produce under the mapping
\item Global non-degeneracy of curves: curves should expand out in different directions in the plane, and should not ``flatten" out across the entire dataset.
\item Local non-degeneracy of curves: curves should expand out in different directions locally as well.
\item Optimal smoothness of curves: curves should exhibit some degree of smoothness to ensure visual clarity (the formal definition of ``optimal" is expressed in Theorem \ref{thm:1}
\item Interpretability of distances: the distances in the space of plane curves should be visually intuitive.
\end{enumerate}

Having specified constraints that ensure forms of these properties hold, we parameterize the linear operators that satisfy the isometry property, global non-degeneracy property, interpretability property, and the optimal smoothness property (Theorem \ref{thm:1}). While the members of this parametric set do not satisfy the local non-degeneracy property, we identify a particular linear operator from this parameterization that enjoys an asymptotic local non-degeneracy property (Theorem \ref{thm:2} provides quantitative bounds). Finally, we observe that local tangent correlations contribute to the ``braiding" of these 3D Andrews plots, which in turn motivates our construction of unit-length space curves (filaments) using the plane curves as data for a system of equations related to the Bishop frame equations \cite{bishop1975there}. 

\paragraph{Organization of this paper.} Section \ref{sec:background} provides background on Andrews plots and related methods, then introduces some notation to provide a technical exposition of ideal properties of 3D Andres plots, and concludes by providing a technical summary of our main results. Section \ref{sec:main} provides a full technical specification of our main results,  Theorem \ref{thm:1} and Theorem \ref{thm:2}. Section \ref{sec:param} provides the proof of Theorem \ref{thm:1}, Section \ref{sec:tour} establishes the proof of Theorem \ref{thm:2}, and Section \ref{sec:filament} describes the non-linear procedure for producing space curves via Bishop frames. Sections \ref{sec:param} and \ref{sec:tour} both conclude by discussing related technical results and ideas. Section \ref{sec:filament} concludes with visualizations of the Boston housing dataset, Wisconsin breast cancer dataset, and a dataset of handwritten digits. Section \ref{sec:conclusion} concludes the paper with a discussion of interesting ramifications, caveats, and potential future directions. 

\section{Background and preliminaries}\label{sec:background}

 2D scatterplots derived from $d$-dimensional data must generally compromise distances or similarities between $d$-dimensional data points. Practically, this means that visual associations in a scatterplot may be misleading. If a dataset $\{x_n\}_{n=1}^N\subset\R^d$ is mapped to $\R^2$ ($x_n\mapsto \Phi(x_n)$), it is relatively straightforward to ensure that there is a reasonably small constant $\delta\geq 0$ such that
\begin{align}
\Vert \Phi(x_n)  - \Phi(x_m)\Vert \leq (1+\delta)\Vert x_n-x_m\Vert \text{ for all } n,m\in[N],
\end{align}
where $\Vert \cdot\Vert$ denotes Euclidean norms and $[N]=\{1,2,\ldots, N\}$ is the $N$-set. Therefore, if two points are close in the data space, then they must be visibly close in the scatterplot. Equivalently, if two points are far away in the scatterplot, then the points must be far away in the data space. 

While we would also like to say that ``if points are close in the scatterplot, then they are close in the data space", this property fails dramatically. In technical terms, bounds of the form 
\begin{align}
\Vert x_n-x_m\Vert\leq (1+\delta^\prime)\Vert \Phi(x_n)  - \Phi(x_m)\Vert \text{ for all } n,m\in[N],
\end{align}
often require a large $\delta^\prime\geq0$. This means that close points in the scatterplot may come from distant points in the data space, and so scatterplots may mislead our visual intuition.

A linear map from $\R^d$ to $\R^2$ which minimizes mean distorting of vectors is obtained by considering the first two principal components of a dataset, but this map give little control over $\delta^\prime$ in the bound above, and these maps also suffer from robustness issues. Many non-linear embedding methods (often predicated on preserving only certain distances) have been introduced to mitigate distortion. Nonlinear multidimensional scaling \cite{kruskal1978multidimensional} attempts to optimize embeddings of datasets to mitigate distortion of distances, stochastic neighborhood embedding (t-SNE) \cite{van2008visualizing} seeks to preserve similarities according to a probabilistic model, \cite{brand2003unifying} indicates that spectral embedding methods preserve angles, and UMAP \cite{mcinnes2018umap} applies category theory, fuzzy logic, and Laplacian eigenmaps to provide approximations to manifolds. Despite the flexibility provided by these methods, they still distort Euclidean distances when the data space dimension $d$ and the number of data samples $N$ is large.

\subsection{Distance-preserving data visualizations}

Instead of mapping each data point to a single point in 2D or 3D, modern computer graphics allows us to visualize maps that send each data point to collections of points. This collection of points could be discrete, or the collection could constitute a function's graph. In particular, the distance in the data space may be completely preserved by a suitable distance between collections of points in low-dimensional space, which strengthens the connection between the data space and the visualization space. 

 In general, the available methods trade visual complexity for data fidelity by merging multiple scatterplots from various projections. These methods include matrix of scatterplots (also known as pairs plots or draftsman's plot) \cite{hartigan1975printer, emerson2013generalized, tukey1981graphical}, parallel coordinates plots/radar charts \cite{inselberg1985plane, von1877gesetzmassigkeit}, Andrews plots \cite{andrews1972plots}, projection pursuit \cite{friedman1974projection}, and tours \cite{asimov1985grand}. The first three methods yield ``static" visualizations that are fixed for all time, while tours provide ``dynamic" visualizations that consist of ``movies" of scatterplots. 

\paragraph{Matrix of scatterplots.} An example is provided by Figure \ref{fig:irisScatterplots} for the Iris dataset \cite{anderson1936species} (a 4D dataset of metric measurements of various iris specimen). Matrix of scatterplots visualizations seemingly retain all information about the dataset, but suffer from an identifiability problem between points in each of the individual scatterplots. This becomes a severe issue when we consider dozens of dimensions since information about a single data point is scattered across the visualization. In fact, the pieces of visual information pertinent to a particular data point are spread uniformly throughout the plot as the dimension increases. Matrix of scatterplots visualizations also introduce bias by only considering a small sample of projections. 

\begin{figure}[ht]
\centering
\includegraphics[scale=0.5]{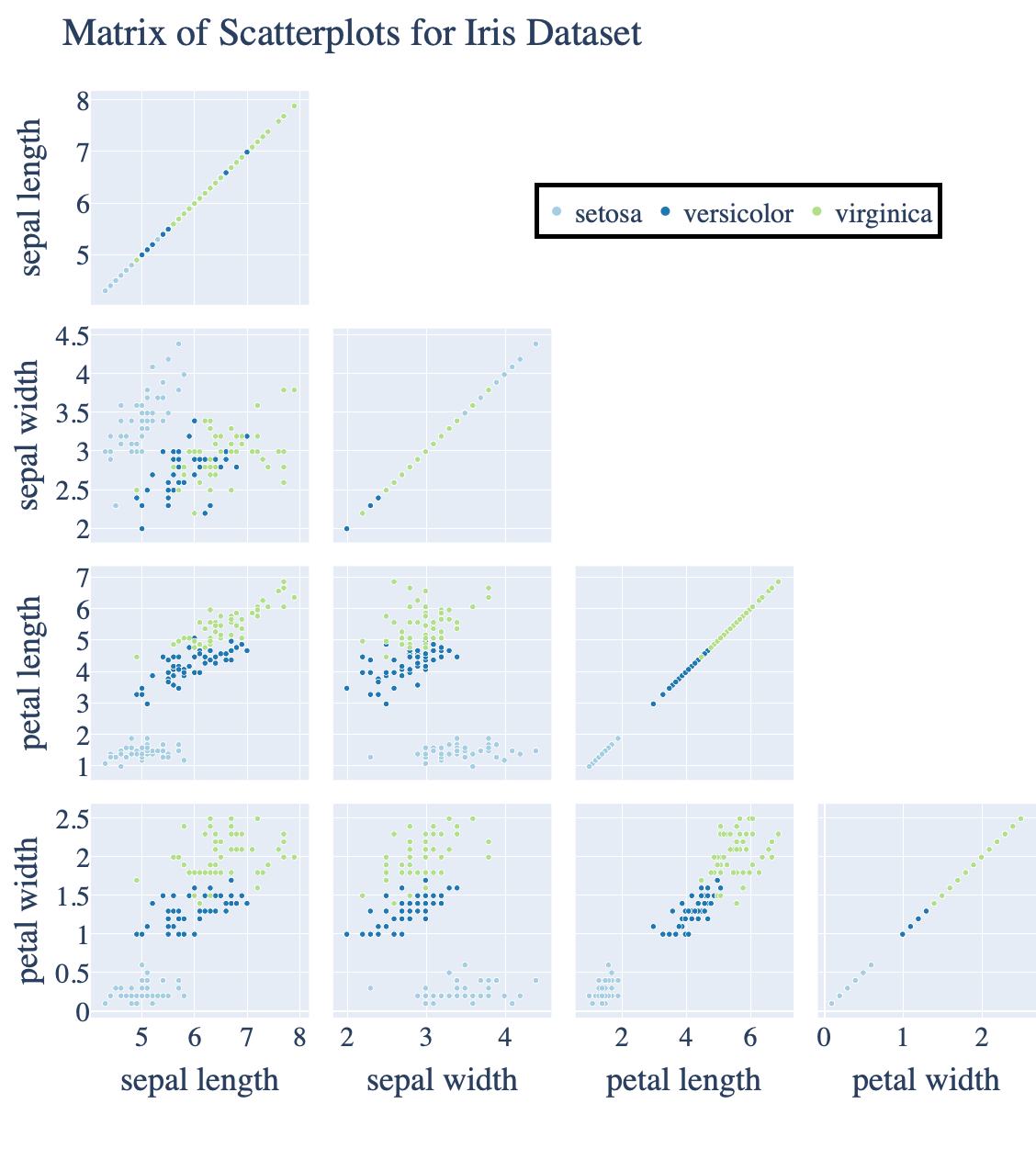}
\caption{Matrix of scatterplots for the Iris dataset.}
\label{fig:irisScatterplots}
\end{figure}

\paragraph{Parallel coordinates.} 
Parallel coordinates plots (Figure \ref{fig:irisParallel}) and radar charts solve the connectivity issue of matrix of scatterplots, but the interpolation approach to identifying individual data points spreads leads to ``visual complexity" in that information about a single data point is spread out substantially over the entire viewing plane. This aesthetic deficiency also lead to practical difficulties in discerning relationships between variables or clusters of data points. Chen et al. \cite{chen2011stringing} consider a procedure for ordering parallel coordinates to minimize these aesthetic issues. Moreover, the graphs consist of many lines in 2D that overlap, making it difficult to pick out any single graph as the number of graphs grows. So called ``small multiples" (see \cite{tufte1990envisioning}) visualizations  offer a partial remedy to this clutter issue, but this sacrifices simultaneous, localized visualization of the entire dataset. 
\begin{figure}[h]
\centering
\includegraphics[scale=0.33]{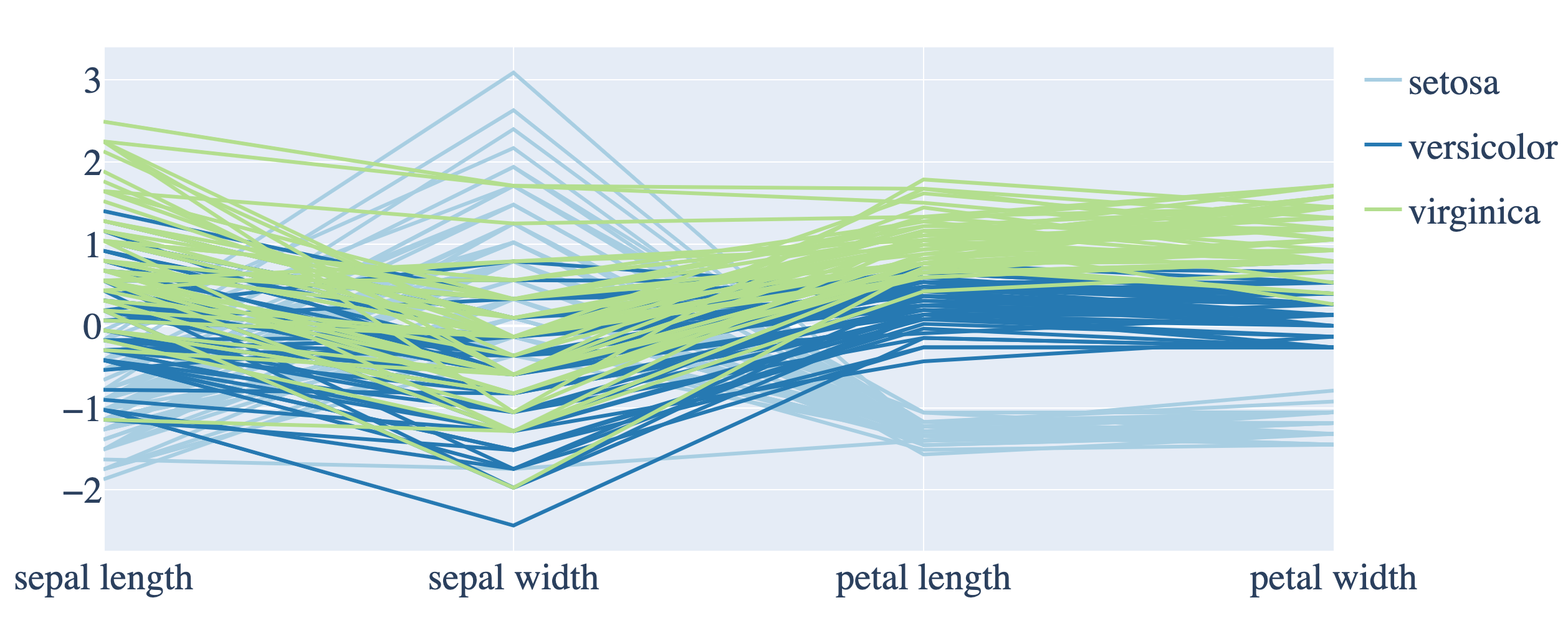}
\caption{Parallel coordinates plot for the Iris dataset.}
\label{fig:irisParallel}
\end{figure}

\paragraph{Projection pursuit and tours.} Grand tours \cite{asimov1985grand} are ``smooth" movies of scatterplots of data under a sequence of projections that vary over a Stiefel manifold. This provides an extreme remedy to the bias produced in the matrix of scatterplots visualizations. This strength leads to two weaknesses: the volume of information overwhelms the viewer and the ``dynamic" nature of the visualization forces the viewer to rely on visual memory and tracking to discern patterns. 

Projection pursuit \cite{friedman1974projection} solves the issue of overwhelming information from the grand tour by selecting ``interesting" projections of the data. \cite{cook1995grand} discusses how to integrate this with the grand tour to obtain visualizations. However, this method is still ``dynamic" as well as computationally intensive.

\paragraph{Andrews plots.} Andrews plots provide a smooth alternative to parallel plots, but may sacrifice the visual identifiability of individual coordinates. Standard Andrews plots are defined by $\Phi:\mathbb{R}^d\to L^2([0,1])$ satisfying
\begin{align}
\Phi[x](t) = x_1 + x_2\cos(2\pi t) + x_3 \sin(2\pi t) + x_4\cos(4\pi t) + x_5\sin(4\pi t) + \cdots
\end{align}
That is, $\Phi[x]$ is a trigonometric polynomial using the lowest $d$ frequencies. \cite{andrews1972plots} justifies using principal component scores\footnote{Slight modifications of the theory presented in our paper indicate that using these PCA scores produces (on average, over the dataset) optimally smooth functions.} instead of the original $x_j$'s by indicating that ``low frequencies are more readily seen". Note that all of the principal components are used, so the plots are a lossless representation of the original dataset. Figure \ref{fig:irisAndrews} illustrates Andrews plots after standard transformations of the Iris dataset. 

\begin{figure}[h]
\centering
\includegraphics[scale=0.35]{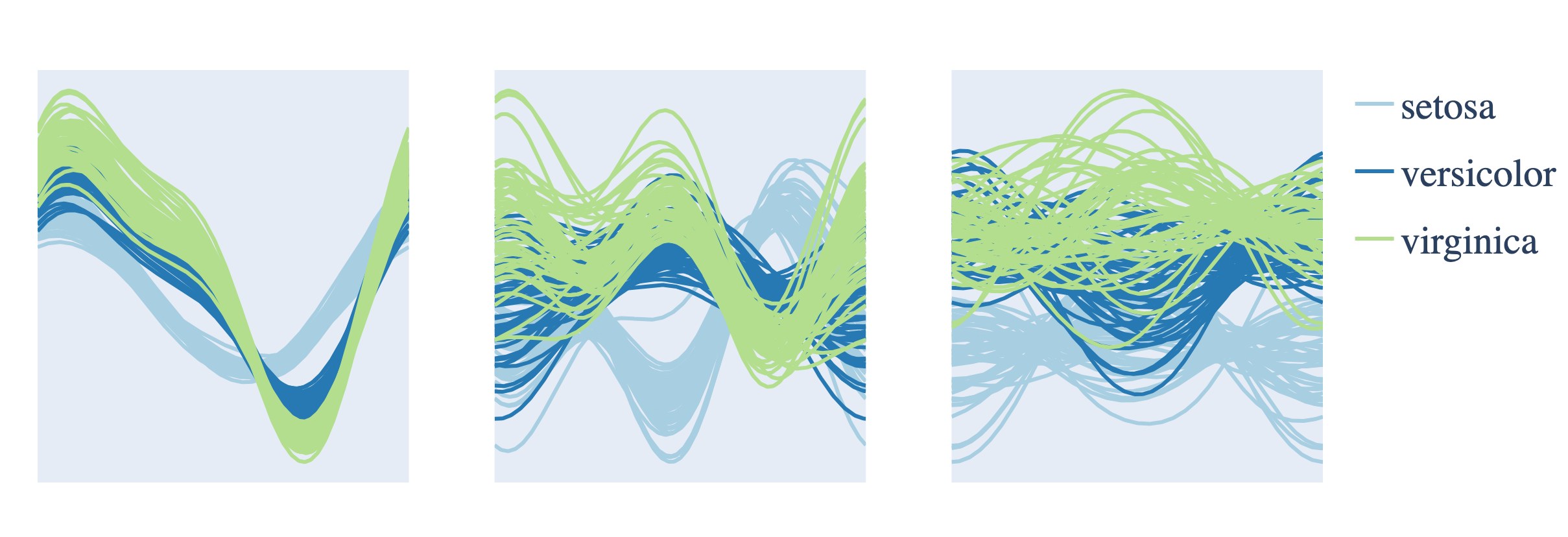}
\caption{Andrews plots for the Iris dataset. The left plot provides the Andrews plot for the raw coefficients from the dataset, the center plot provides the Andrews plot after the coefficients are standardized, and the right plot uses the PCA scores after standardization. While the overall clutter of these images appears to increase from left to right, it is also easier to pick out individual curves in the rightmost plot as opposed o the leftmost plot.}
\label{fig:irisAndrews}
\end{figure}

Since Andrews original paper, many authors introduced variations and interpretations of Andrews plots. A recent perspective on Andrews plot is given in \cite{moustafa2011andrews}. Khattree and Naik \cite{khattree2002andrews} review several different ways to combine $\sin$ and $\cos$ functions, as well as various choices for frequency multipliers. Wegman and Shen \cite{wegman1993three} consider multivariate Andrews plots. Embrechts and Herzberg \cite{embrechts1991variations} study variations of Andrews plots involving basis functions beyond the Fourier basis. Several papers extend Andrews plots to map data points to parametric surfaces. Koziol and Hacke \cite{koziol1991bivariate} study the case where data comes in pairs and a provide embeddings from data space to spaces of smooth surfaces from this framework. Garcia et al. \cite{garcia2004extension, garcia2005visualization} consider visualizations resulting from slices of 3D graphs of bivariate functions related to Andrews plots.

Wegman and Shen \cite{wegman1993three} creates a fast approximate grand tour structure by exploiting orthogonality -- however this projective property does not maintain the full orthogonality property of the basis functions in each variable, so there are redundancies in the components of the curves. Their system is
\begin{align}
\Phi(t) = \sqrt{\frac{2}{d}} \begin{pmatrix}
\sin(\lambda_1 t) & \cos(\lambda_1 t) & \cdots & \sin(\lambda_{d/2} t) & \cos(\lambda_{d/2} t)\\
\cos(\lambda_1 t) & -\sin(\lambda_1t) & \cdots & \cos(\lambda_{d/2} t) & -\sin(\lambda_{d/2} t)
\end{pmatrix}
\end{align}
where $\lambda_1,\ldots, \lambda_{d/2}$ are linearly independent over the rationals. This also sacrifices the isometry property. Instead, a ``long term average" isometry property holds.  Garcia et al. \cite{garcia2005visualization} note that this is related to a basis in derivatives of periodic curves, and observe that the system of projections encourages clusters to exhibit ``flocking" behavior that data clusters may exhibit.

As indicated in Andrews original paper, the plots constitute a \emph{linear isometry} from $\R^d$ to $L^2([0,1])$. That is, $\Phi[ax+by]=a\Phi[x] + b\Phi[y]$ for all $x,y\in\R^d$ and $a,b\in\R$, and
\begin{align}
\Vert \Phi[x-y]\Vert_{L^2} = \Vert x-y\Vert\text{ for all }x,y\in\R^d.
\end{align}
Here, $\Vert \cdot \Vert_{L^2}$ is the standard $L^2$ norm on $L^2([0,1])$, and we see that Andrews plots preserve Euclidean distances without distortion in the $L^2$ metric. This isometry property follows because the Fourier basis constitutes an orthonormal basis of $L^2([0,1])$. This means that, unlike scatterplots in $\R^2$, Andrews plots ensure that proximity in the visualization space implies proximity in the data space.

Unfortunately, $L^2$ distances between functions are not visually intuitive. One may view $\pi \Vert f\Vert_{L^2}^2$ as a volume of revolution, or an experienced electrical engineer may be able to understand $\Vert f\Vert_{L^2}^2$ as the energy of a voltage signal, but in general this quantity is not intuitive for most graphs of functions. On the other hand, the standard bounds
\begin{align}
\Vert \Phi[x-y]\Vert_{L^1} \leq \Vert \Phi[x-y]\Vert_{L^2} \leq \Vert \Phi[x-y]\Vert_{L^\infty}
\end{align}
indicate that the net area between the Andrews plots is a lower bound, and the maximum absolute difference between the curves is an upper bound for the distance between $x$ and $y$ when $\Phi$ is a linear isometry. This upper bound is visually intuitive because we may just form a uniform band around an Andrews plot to obtain a neighborhood. The lower bound is less useful for visual intuition because areas are judged less accurately than lengths (see \cite{stewart2010examination} for a summary of the relevant psychological literature). 

Despite numerous variations on Andrews original idea, they still exhibit an issue with clutter; graphs are 1D curves plotted in 2D, so the presence of many nearby curves obscures information as the number of data points increases. Brushing techniques \cite{becker1987brushing} are sometimes used to overcome this by highlighting collections of curves, but this technique may be applied to any visualization. By starting with a procedure that mitigates visual clutter, brushing techniques may provide even more visual clarity.

\subsection{Ideal properties of 3D Andrews plots}

In order to provide better spatial separation of curves for visualization, we consider graphs of planar parametric curves plotted in 3D with a third axis for the ``time" parameter. Standard ``drag-to-rotate" and ``zoom" interfaces allow us to fully navigate these types of visualizations. Ideally, these curves also retain all the desirable properties of Andrews plots:

\begin{enumerate}
\item The map from the Euclidean data space to plane curves (identified with $\left(L^2([0,1])\right)^2$) is a linear isometry.
\item The linear isometry takes data to smooth curves.
\item The linear isometry property has accessible visual interpretations.
\item Displaying the curves is computationally inexpensive.
\end{enumerate}
Subtle issues arise when we attempt to maintain these properties for planar parametric curves. We now identify five important properties that ensure visual appeal. By combining these properties, we obtain a constrained optimization problem over an infinite dimensional space of linear transformations.

\paragraph{Global non-degeneracy of curves and isotropic isometries.} First, a map of the form
\begin{align}
\Phi[x](t) = \begin{pmatrix}
x_1 + x_2\sqrt{2}\cos(2\pi t) + x_3\sqrt{2}\sin(2\pi t) + \cdots\\
0
\end{pmatrix}
\end{align}
is a linear isometry, but the plots of the resulting curves are confined to a two dimensional subspace of $\R^3$. Such degeneracies contradict the intent of mapping into 3D in the first place. To restrict this type of degenerate behavior, we introduce the \emph{isotropic isometry} condition for a linear map from the data space to the space of planar curves: $\Phi:\R^d\to \left(L^2([0,1])\right)^2$ is such that, for any $u\in\R^2$ with $\Vert u\Vert=1$,
\begin{align}
\Vert u^T\Phi[x]\Vert_{L^2} = \Vert x\Vert\text{ for all } x\in\R^d.
\end{align}
That is, for any unit vector $u\in\R^2$, the projection $x\mapsto u^T\Phi[x]$ is an isometry from $\R^d$ to $L^2([0,1])$.

\paragraph{Local non-degeneracy of curves and the projective tour property.} While the isotropic isometry condition prevents ``global" degeneracies of curves, it is desirable to also prevent a ``local" degeneracy of the curves as well. Suppose the linear map $\Phi(t):\R^d\to\R^2$ has rank $1$ for all $t\in[0,1]$. Then the ``skinny" singular value decomposition of $\Phi(t)$ is $\Phi(t) = u(t)\sigma(t) v(t)^T$ where $u(t)\in \R^2$, $\sigma(t)\in\R$, and $v(t)\in\R^d$ for all $t\in[0,1]$. Moreover, $u$ satisfies (a) $\Vert u(t)\Vert=1$ for all $t\in[0,1]$, and (b) for each $x\in\R^d$ there is a function $\gamma_x$ ($\gamma_x(t)=\sigma(t)v(t)^Tx$ for all $t\in[0,1]$) with $\Phi[x](t) = \gamma_x(t)u(t)$ for all $t\in[0,1]$. 

On the other hand, if $\Phi(t)$ has rank $2$ for all $t\in[0,1]$, then the singular value decomposition has the form
\begin{align}
\Phi(t) = \begin{pmatrix} u_1(t) & u_2(t)\end{pmatrix} \begin{pmatrix} \sigma_1(t) & 0 \\ 0 & \sigma_2(t)\end{pmatrix} \begin{pmatrix}v_1(t)^T\\ v_2(t)^T\end{pmatrix}
\end{align}
where $u_1(t), u_2(t)\in \R^2$ are orthonormal, $\sigma_1(t),\sigma_2(t)\in\R$ are strictly positive, and $v_1(t),v_2(t)\in\R^d$ for all $t\in[0,1]$. Therefore, there are functions $\gamma_{x,1},\gamma_{x,2}$ ($\gamma_{x,i}(t) = \sigma_i(t)v_i(t)^Tx$ for all $t\in[0,1]$ and $i=1,2$) such that
\begin{align}
\Phi[x](t) = \gamma_{x,1}(t)u_1(t)+\gamma_{x,2}(t)u_2(t)\text{ for all }t\in[0.1].
\end{align}
In the rank $1$ case, we see that $\Vert \Phi[x]-\Phi[y]\Vert_{L^2}$ is mediated by only a single function $\gamma_x-\gamma_y$, whereas this difference is mediated by two functions $\{\gamma_{x,k}-\gamma_{y,k}\}_{k=1}^2$ in the rank $2$ case. This means that the rank $2$ case allows curves to diverge along as many directions as possible as the curves evolve in the plane. Therefore, to encourage this diverse directional divergence behavior, our final criteria is that $\Phi$ approximates a projective tour: there exists a $c>0$ such that $c\Phi(t)$ is a rank $2$ projection (i.e. it is a co-isometry or has two non-zero singular values equal to $1$) for all $t\in[0,1]$. 

\paragraph{Smooth curves and mean quadratic variation.} The set of isotropic isometries from data space to spaces of plane curves with the form $\Phi:\R^d\to \left(L^2([0,1])\right)^2$, is vast and a generic choice of $\Phi$ from this set produces noisy, high-frequency curves. To ensure the production of smooth curves, we choose $\Phi$ which minimizes the \emph{mean quadratic variation}
\begin{align}
\text{MQV}(\Phi, X)=\frac{1}{N} \sum_{n=1}^N \left\Vert \frac{d\Phi[x_n]}{dt}\right\Vert_{L^2}^2
\end{align} 
where $X=\begin{pmatrix} x_1 & x_2 & \cdots & x_N\end{pmatrix}\in\R^{d\times N}$ is a data matrix for a dataset $\{x_j\}_{j=1}^N\subset\R^d$, and $\Vert \cdot \Vert_{L^2}$ is an $L^2$ norm for vector-valued functions over $[0,1]$:
\begin{align}
\left\Vert \frac{d\Phi[x_n]}{dt}\right\Vert_{L^2}^2 = \int_0^1\left\Vert \frac{d\Phi[x_n]}{dt}(t)\right\Vert^2\:dt.
\end{align}
By minimizing the mean quadratic variation over a dataset, we promote curves exhibiting smoothness (i.e. functions with bounded support in the Fourier domain). 

\paragraph{Interpretability of isometries for spaces of derivatives.} Finally, we consider the interpretation of the $L^1$ and $L^\infty$ bounds in the context of these planar parametric curves. While the $L^\infty$ bound is relatively easy to interpret, the $L^1$ bound is not visually intuitive. However, if we assume that the isometry $\Phi$ takes data to \emph{derivatives} of smooth curves, then the $L^1$ bound describes the lengths of curves and the $L^\infty$ bound applies to the instantaneous velocities of curves. 

\paragraph{Removal of visual bias and closed curves.} Periodic curves with period $1$ may be written as
\begin{align}
\gamma(t) = c + \tilde\gamma(t)
\end{align}
where $c\in\R^2$, $\tilde\gamma(t+1) = \tilde\gamma(t)$, and $\int_0^1 \tilde\gamma(t)\:dt=0$ if the curve is integrable on $[0,1]$. Consequently, integrals of periodic planar parametric curves may be separated into affine and periodic components:
\begin{align}
\int_0^s \gamma(t)\:dt = sc + \int_0^s\tilde\gamma(t)\:dt
\end{align}
A linear isometry from $\R^d$ to $\left(L^2([0,1])\right)^2$ induces a linear isometry to the coefficients of the affine component of these maps. Inclusion of the affine components creates a visual impact dominated by information equivalent to two scatterplots, which is necessarily lossless. To avoid this visual bias, we restrict the images of isometries to periodic functions on $[0,1]$. 

\subsection{Summary of main results}

While minimizing the mean quadratic variation subject to the isotropic isometry property and the projective tour property presents a substantial computational challenge, it is possible to characterize and parameterize the minimizers which satisfy just the isotropic isometry property, and then we may extract a minimizer from this parameterization that approximately satisfies the tour property. 

Moreover, this minimizer and the resulting curves only requires the computation of a singular value decomposition of the data matrix. As such, our methods are (relatively) computationally inexpensive. We now summarize our main results:
\begin{enumerate}
\item We parameterize the set of minimizers of the mean quadratic variation subject to the isotropic isometry condition (Theorem \ref{thm:1}).
\item In general, the set of minimizers has $d$ degrees of freedom, and we exhibit a particular choice of minimizer that admits a projective tour property in an asymptotic sense (Theorem \ref{thm:2}). This choice is motivated by recent results on quadratic Gauss sums.
\end{enumerate}
It should be noted that one may reformulate the proof of Theorem \ref{thm:1} to demonstrate that PCA scores provide Andrews plots that minimize the mean quadratic variation over a dataset subject to an isometry condition. The rightmost plot in Figure \ref{fig:irisAndrews2D} illustrates the resulting plots for the Iris dataset.

When plotting the resulting graphs of curves, we observe that curves tend to follow each other due to the fact that the tangent vectors of the graphs always agree on the first component. To further declutter the 3D plots, we consider mapping from planar curves ($\left(L^2([0,1])\right)^2$) to the space of arc-length parameterized curves in $\R^3$ with total length $1$ by imposing the condition that the derivative of the tangent of this curve satisfies
\begin{align}
d{\bf T}(t) = \phi_1(t) {\bf N}_1(t) + \phi_2(t){\bf N}_2(t)
\end{align}
where ${\bf T}, {\bf N}_1, {\bf N}_2$ constitute a moving frame determined by the functions $\phi_1,\phi_2\in L^2([0,1])$. This type of moving frame is related to the Frenet-Serret moving frame, and has been dubbed the ``Bishop frame" of a curve.  In particular, to produce a filament $\gamma_i$ from a data point $x_i$, we let $\phi_j[x]$ denote the $j$th component function of the map $\Phi[x_i]$ and numerically solve a linear first order matrix ODE. Therefore, our $\R^2$ curves give rise to $\R^3$ curves via numerical solutions to a system of equations that produces Bishop frames \cite{bishop1975there}. We also show that the isometry condition may be interpreted in terms of ``relative" curvatures. This interpretation follows from the fact that a curve satisfying the above formula has the curvature function $\kappa(t) = \sqrt{\phi_1^2(t)+\phi_2^2(t)}$). While the resulting curves now involve a non-linear transformation of the data that exhibits anisotropic evolution of data curves, the resulting plots transform the ``braided" 3D Andrews plots into ``bushy" plots that exhibit (empirically) better visual separation. These filament plots are illustrated in Figure \ref{fig:irisFilament}.

\section{Main results}\label{sec:main}

Our results hold over spaces of square integrable functions. Let $L^2([0,1])$ denote the Hilbert space of square integrable (in the Lebesgue sense) functions with norm $\Vert \cdot\Vert_{L^2}$ and inner product $\langle\cdot,\cdot\rangle_{L^2}$, and set
\begin{align}
\mc{H}^1([0,1]) = \left\{f\in L^2([0,1]): \frac{df}{dt}\text{ is defined almost everywhere, } \frac{df}{dt}\in L^2([0,1])\right\}
\end{align} 
In other words, $\mc{H}^1([0,1])$ is the Sobolev space of square integrable functions over $[0,1]$ with square integrable derivatives over $[0,1]$. Since the decorations of $\mc{H}^1([0,1])$ never change, we simply write $\mc{H}$ for the remainder of the paper.

We identify the space of parameterized linear functionals as the $d$-fold Cartesian product of $d$ copies of $\mc{H}$ as $\mc{H}^d\subset \left(L^2([0,1])\right)^d$. In particular, $\phi\in \mc{H}^d$ if and only if $\phi:[0,1]\to \R^d$ has component functions $\phi_k\in\mc{H}$ for $k=1,\ldots, d$. Note that $\mc{H}^d$ inherits the (incomplete) norm
\begin{align}
\Vert \phi\Vert_{L^2} = \sqrt{\int_0^1 \Vert\phi(t)\Vert^2\:dt}
\end{align}
from $\left(L^2([0,1])\right)^d$, where $\Vert \cdot\Vert$ is the standard Euclidean norm on $\R^d$.

We define the space of parameterized linear mappings from the data space $\mathbb{R}^d$ to $\R^2$ by $\mc{H}^{2\times d}\subset \left(L^2([0,1])\right)^{2\times d}$ so that $\Phi\in\mc{H}^{2\times d}$ if and only if $\Phi:[0,1]\to \R^{2\times d}$ has component functions $\phi_{j,k}\in\mc{H}$ for all $j=1,2$ and $k=1,\ldots, d$. Here $\R^{2\times d}$ denotes the space of all $2$ by $d$ matrices with real entries. Note that $\mc{H}^{2\times d}$ inherits the norm
\begin{align}
\Vert \Phi\Vert_{L^2} = \sqrt{\int_0^1 \Vert\Phi(t)\Vert^2\:dt}
\end{align}
from $\left(L^2([0,1])\right)^{2\times d}$, where
\begin{align}
\Vert A \Vert = \sqrt{\text{trace}(A^TA)}
\end{align}
is the Frobenius norm of a matrix $A\in\R^{2\times d}$.

For $\Phi\in\mc{H}^{2\times d}$, we call the $\phi_{j,\cdot}^T\in \mc{H}^d$ functions the \emph{component slices} of $\Phi$, we call the $2$ by $d$ matrices 
\begin{align}
\Phi(t) = \begin{pmatrix}
\phi_{1,1}(t) & \phi_{1,2}(t) & \cdots & \phi_{1,d}(t)\\
\phi_{2,1}(t) & \phi_{2,2}(t) & \cdots & \phi_{2,d}(t)\\
\end{pmatrix}
\end{align} 
the \emph{time slices} of $\Phi$, and we call the functions $\phi_{\cdot, k}\in\mc{H}^2$ the \emph{coordinate slices} of $\Phi$. Figure \ref{fig:tensorSlices} illustrates the time slices of such a tensor.

\begin{figure}[h]
\centering
\includegraphics[scale=0.5]{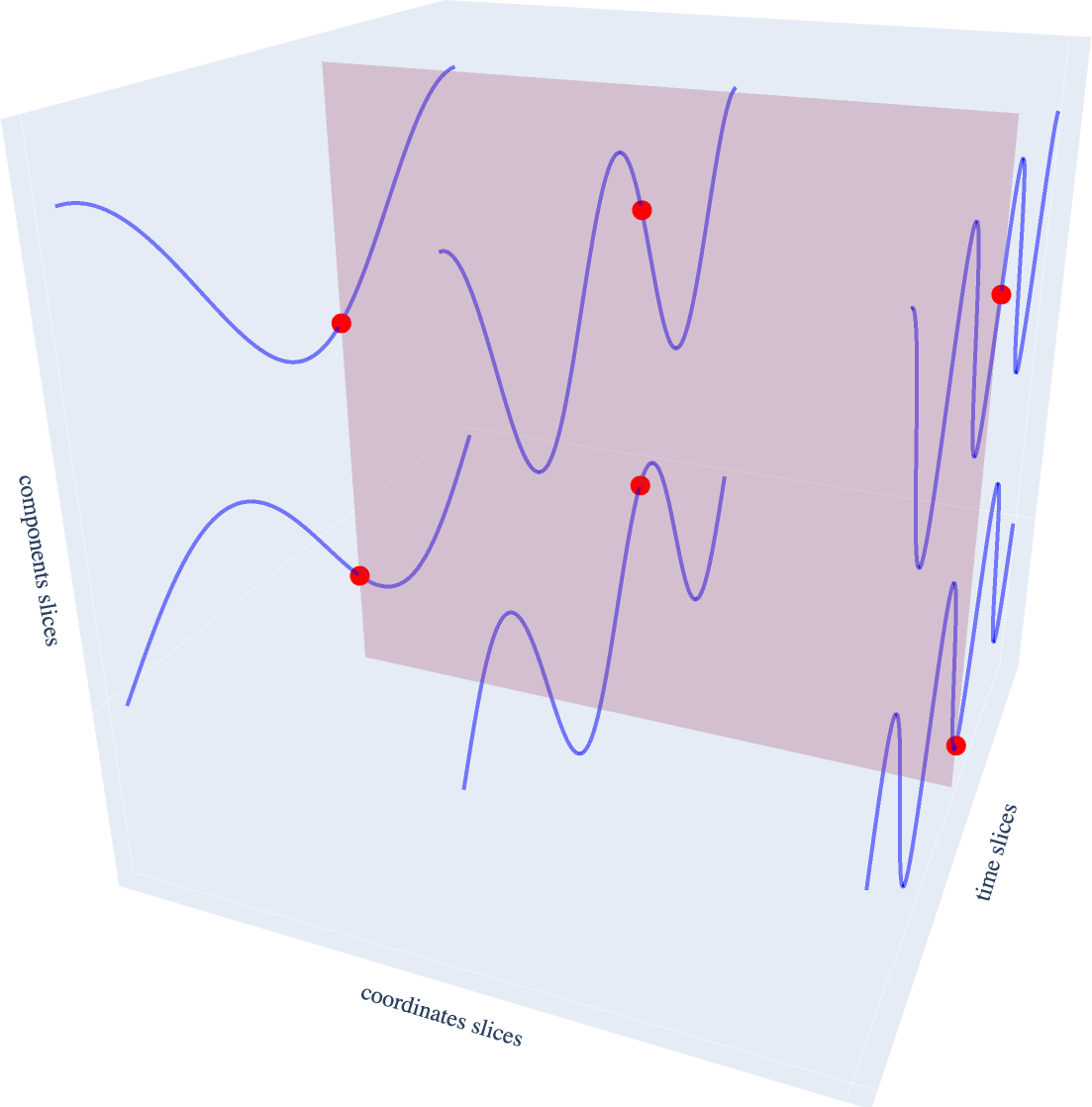}
\caption{An example of ``false" time slices for $\Phi\in\mc{H}^{2\times 3}$. Here there are two possible component slices, three possible coordinate slices, and a continuum of time slices. The ``entries" of a time slice $\Phi(t)$ are illustrated by the red markers.  }
\label{fig:tensorSlices}
\end{figure}

Consider the set of all linear operators $\mc{L}(\R^d,\mc{H}^2)$. These are all bounded (and hence continuous) since the unit ball of $\R^d$ is compact. Just as linear functions between Euclidean spaces are related to matrices, members of $\mc{L}(\R^d,\mc{H}^2)$ may be characterized by matrix-valued functions whose entries are characterized by the next proposition.  

\begin{proposition}\label{prop:1} The operator $\Omega\in \mc{L}(\R^d, \mc{H}^2)$ is linear if and only if there is a $\Phi\in \mc{H}^{2\times d}$ such that $
\Omega[x](t) = \Phi(t) x$ for all $x\in\R^d$ and almost all $t\in[0,1]$.
\end{proposition}

The proof of this proposition is a simple exercise, so we omit it. Based on Proposition \ref{prop:1}, we abuse notation by allowing $\Phi[x]$ to denote the operator induced by $\Phi\in\mc{H}^{2\times d}$ applied to the vector $x\in\R^d$. We define the \emph{quadratic variation} of $f\in\mc{H}^2$ by
 \begin{align}
 \text{QV}(f) = \int_0^1 \left\Vert \frac{d f}{dt}(t)\right\Vert^2\: dt.
 \end{align}
For a dataset $X=\{x_n\}_{n=1}^N\subset\mathbb{R}^d$ and a mapping $\Phi\in\mc{H}^{2\times d}$ the mean quadratic variation is
 \begin{align}
 \text{MQV}(\Phi; X) = \frac{1}{N}\sum_{n=1}^N \text{QV}(\Phi[x_n]).
 \end{align}
This measures the average ``smoothness" of the $\Phi$ mapping over the dataset. The next proposition simplifies the structure of the isotropic isometry condition for such a $\Phi$.

\begin{proposition} The operator $\Phi\in \mc{L}(\R^d,\mc{H}^2)$ is an isotropic isometry if and only if
\begin{align}
\int_0^1\phi_{j,k}(t) \phi_{j^\prime, k^\prime}(t)\:dt = \delta_{(j,k),(j^\prime,k^\prime)}
\end{align}
for all $(j,k),(j^\prime,k^\prime)\in[2]\times[d]$, where
$[2] =\{1,2\}$, $[d]=\{1,\ldots,d\}$, and $\delta_{(k,j),(k^\prime,j^\prime)}$ is a Kronecker delta function. 
\end{proposition}

In other words, the entry functions of the matrix-valued function $\Phi$ form an orthonormal collection in $L^2([0,1])$. The proof of the difficult implication in this proposition is a straightforward application of the parallelogram identity, and we omit the full proof of this proposition due to its simplicity.

We additionally impose the constraint that
\begin{align}
\int_0^1 \Phi[x](t)\:dt = 0\text{ for all } x\in\R^d.
\end{align}
This constraint implies $\Phi$ maps into spaces of discrete derivatives of parameterizations of closed curves. The conditions
\begin{align}
\int_0^1 \phi_{j,k}(t)\:dt = 0\text{ for all }j\in[2],k\in[d]
\end{align}
encode this constraint. Combining this constraint with the isometric isometry constraint, we have that $\Phi$ is a norm-preserving map onto spaces of derivatives of parameterizations of closed curves. Therefore, $\Phi$ induces a map onto spaces of functions such that the norm-difference in the data space is analogous to the energy-difference (in the sense of the energy of a voltage signal) of the path. 

We now define the set $\mc{S}^{2,d}(X)$ that shall ultimately constitute the minimizers of our constrained optimization program. For a matrix $A\in\R^{2\times k}$, we define the \emph{vectorization operator} $\alpha:\R^{2\times k}\to \R^{2k}$ by
\begin{align}
\alpha \begin{pmatrix}
a_{1,1} & a_{1,2} & \cdots &  a_{1,k}\\
a_{2,1} & a_{2,2} & \cdots &  a_{2,k}\\
\end{pmatrix} = \begin{pmatrix}
a_{1,1} \\ a_{1,2} \\ \vdots \\  a_{1,k}\\
a_{2,1} \\ a_{2,2} \\ \vdots \\  a_{2,k}\\
\end{pmatrix}\text{, and note that } \alpha^{-1}\begin{pmatrix} x \\ y\end{pmatrix} = \begin{pmatrix} x^T\\ y^T\end{pmatrix}.
\end{align}
for all $x,y\in\R^{k}$. Using this vectorization operator, we define a group action of the orthogonal group $\mc{O}(2k)$ on $\mc{H}^{2\times k}$ so that when $U\in\mc{O}(2k)$ and $\Psi\in\mc{H}^{2\times k}$, then $U\cdot \Psi \in\mc{H}^{2\times k}$ is defined by
\begin{align}
(U\cdot \Psi)(t) = \alpha^{-1}\left(U\alpha(\Psi(t))\right) \text{ for almost all } t\in[0,1].
\end{align}

We define the functions $c_k, s_k\in \mc{H}$ by $c_k(t) = \sqrt{2}\cos(2\pi k t)$ and $s_k(t) = \sqrt{2} \sin(2\pi k t)$ (note that $\sqrt{2}$ is the normalization constant). For $k_1,k_2\in[d]$ satisfying $k_1\leq k_2$, we let $\Pi_{[k_1,k_2]}$ denote the orbit of
\begin{align}
\Psi_{k_1,k_2}(t) = \begin{pmatrix}
c_{k_1}(t) & \cdots & c_{k_2}(t)\\
s_{k_1}(t) & \cdots & s_{k_2}(t)\\
\end{pmatrix}
\end{align}
under this action by $\mc{O}(2(k_2-k_1+1))$. We note that $\Psi\in\Pi_{[k_1,k_2]}$ if and only if the component functions of $\Psi$ form an orthonormal basis of
\begin{align}
\text{span}\{c_{k_1},\ldots, c_{k_2},s_{k_1},\ldots, s_{k_2}\}\subset L^2([0,1]).
\end{align}
When $k\in[d]$, we simply write $\Pi_k = \Pi_{[k,k]}$ and note that 
and note that $\Pi_k$ is the orbit of the curve $
\begin{pmatrix}
c_k &
s_k
\end{pmatrix}^T \in \mc{H}^2$
under the left-multiplication action of $\mc{O}(2)$.

 Let $X\in\R^{d\times N}$, suppose $X=U\Sigma V^T$ is a singular value decomposition of $X$ (where the diagonal entries of $\Sigma$, or singular values, are in the standard, non-increasing order), and that the partition $\{[k_{q,\min},k_{q,\max}]\}_{q=1}^p$ of $[d]$ satisfies 
\begin{align}
\sigma_k = \sigma_{k_{q,\min}} \text{ if and only if } k\in[k_{q,\min},k_{q,\max}].
\end{align}

\begin{definition}
Suppose $X\in\R^{d\times N}$ has the SVD $X=U\Sigma V^T$. Then $\Phi \in \mc{S}^{2,d}(X)$ if and only if $\widetilde{\Phi}(t) = \Phi(t) U$ satisfies 
\begin{align}
\widetilde{\Phi}_{[k_{q,\min},k_{q,\max}]} = \begin{pmatrix}
\widetilde{\phi}_{1, k_{q,\min}} & \cdots & \widetilde{\phi}_{1, k_{q,\max}}\\
\widetilde{\phi}_{2, k_{q,\min}} & \cdots & \widetilde{\phi}_{2, k_{q,\max}}\\
\end{pmatrix} \in \Pi_{[k_{q,\min},k_{q,\max}]}
\end{align}
for all $q\in[p]$.
\end{definition} 

With this definition in hand, we now state the form of the optimization program we have derived. Recalling that
\begin{align}
\text{MQV}(\Phi, X)=\frac{1}{N} \sum_{n=1}^N \left\Vert \frac{d\Phi[x_n]}{dt}\right\Vert_{L^2}^2,
\end{align} 
our optimization program is
\begin{align}
\min_{\Phi\in\mc{H}^{2\times d}} \text{MQV}(\Phi;X) \label{mmqv}
\end{align}
subject to the constraints
\begin{align}
\int_0^1\phi_{k,j}(t) \phi_{k^\prime, j^\prime}(t)\:dt = \delta_{(k,j),(k^\prime,j^\prime)}.
\end{align}
for all $(k,j),(k^\prime,j^\prime)\in[2]\times[d]$ and
\begin{align}
\int_0^1 \phi_{k,j}(t)\:dt = 0.
\end{align}
for all $(k,j)\in[2]\times[d]$. We are now in a position to state our main result.

\begin{theorem}\label{thm:1}
The system $\Phi\in\mc{L}(\R^d,\mc{H}^2)$ solves Program \ref{mmqv} if and only if $\Phi\in\mc{S}^{2,d}(X)$.
\end{theorem}

This theorem completely specifies the degrees of freedom for the solutions to the minimization problem. The following corollary considers the simplification when the singular values of $X$ are distinct.

\begin{corollary}
If the $X=U\Sigma V^T\in\R^{d\times N}$ has distinct singular values, then $\Phi\in\mc{S}^{2,d}(X)$ if and only if there is a collection $\{Q_k\}_{k=1}^d\subset \mc{O}(2)$ such that $\widetilde{\Phi}=\Phi\cdot U$ has coordinate slices
\begin{align}
\widetilde{\phi}_{\cdot, k} = Q_k \begin{pmatrix}
c_k\\
s_k
\end{pmatrix}
\end{align}
\end{corollary}

\begin{example}
Figure \ref{fig:tensorSlices} illustrates the matrix-valued function
\begin{align}
\Psi(t) = \begin{pmatrix}
c_1(t) & c_2(t) & c_3(t)\\
s_1(t) & s_2(t) & s_3(t)
\end{pmatrix}.
\end{align}
Given a data matrix $X\in\R^{3\times N}$ with distinct singular values and an SVD $X=U\Sigma V^T$, then $\Phi$ is a minimizer of Program \ref{mmqv} if and only if there are $2$ by $2$ rotation matrices $Q_1$, $Q_2$, and $Q_3$ such that
\begin{align}
\Phi(t) = \begin{pmatrix}
Q_1 \begin{pmatrix} c_1(t)\\ s_1(t)\end{pmatrix} & Q_2 \begin{pmatrix} c_2(t)\\ s_2(t)\end{pmatrix} & Q_3 \begin{pmatrix} c_3(t)\\ s_3(t)\end{pmatrix}
\end{pmatrix} U^T
\end{align}
for almost all $t\in[0,1]$. 
\end{example}

In the generic case, $X=U\Sigma V^T$ has distinct singular values and $\mc{S}^{2\times d}(X)$ is invariant under the action of $\mc{O}(2)^d$ that applies to each column after multiplication by $U^T$. We now exploit these degrees of freedom to push the time slices $\Phi(t)$ towards projections to obtain an approximate tour property. We let
\begin{align}
U(s) = \begin{pmatrix}
\cos(2\pi s) & -\sin(2\pi s)\\
\sin(2\pi s) & \cos(2\pi s)
\end{pmatrix}.
\end{align}

\begin{theorem}\label{thm:2}
Fix $\Phi\in\mc{S}^{2\times d}(I_d)$ by specifying $U_k = U(k^2/4d)$ for $k\in[d]$ and setting
\begin{align}
\phi_{\cdot, k} = U_k \begin{pmatrix}
c_k\\
s_k
\end{pmatrix}\text{ for all }k\in[d]
\end{align}
Then the scaled time slices $\sqrt{\frac{1}{d}}\Phi(t)$ all have singular values in the interval
\begin{align}
\left[\sqrt{1 - \left( \frac{4}{\sqrt{d}}+\frac{3}{2d}+\frac{1}{d^2}\right)},\:\sqrt{1 + \left( \frac{4}{\sqrt{d}}+\frac{3}{2d}+\frac{1}{d^2}\right)}\right].
\end{align}
for all $t\in[0,1]$. 
\end{theorem}

\begin{example}
Figure \ref{fig:qps} illustrates the behavior of the singular values of
\begin{align}
\frac{1}{2}\Psi(t) = \frac{1}{\sqrt{2}}\begin{pmatrix}
\cos(2\pi t) & \cos(4\pi t) & \cos(6\pi t) & \cos(8\pi t)\\
\sin(2\pi t) & \sin(4\pi t) & \sin(6\pi t) & \sin(8\pi t)
\end{pmatrix}
\end{align}
and $\Phi(t)$ obtained through quadratic phase shifts of the columns of $\Psi$ suggested by Theorem \ref{thm:2}.
\end{example}

\begin{figure}[h]
\centering
\includegraphics[scale=0.35]{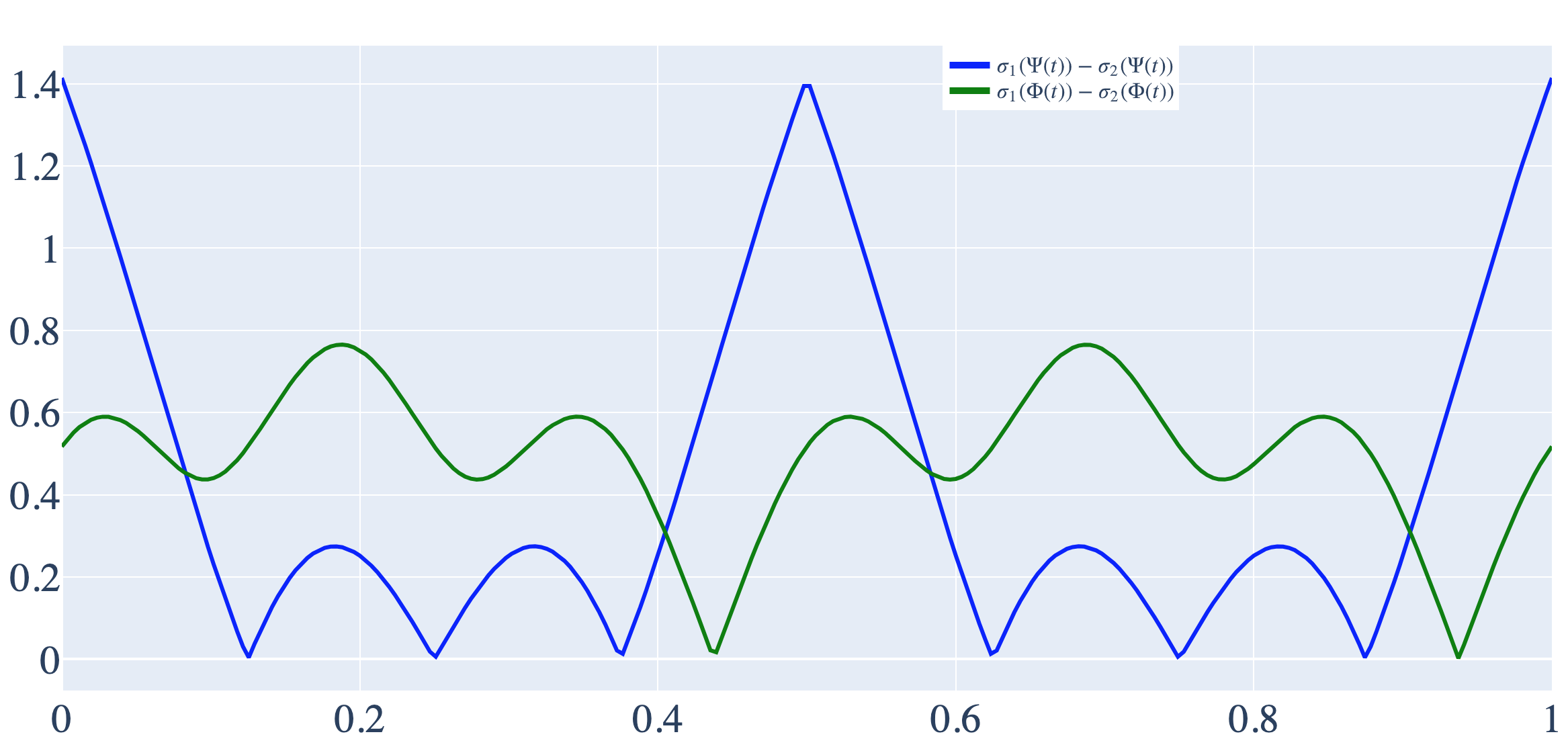}
\caption{Singular value gaps before and after quadratic phase shifts for minimizers $\Psi$ and $\Phi$ from Theorem \ref{thm:1} with $d=4$. Note that a gap of $\sqrt{2}$ (attained by $\Psi$) indicates a completely degenerate rank). }
\label{fig:qps}
\end{figure}

We note that the basis proposed by Khattree et al. \cite{khattree2002andrews} is obliquely related to this choice (for 1D Andrews plots). Their suggested basis consists of rotations of the orthogonal pairs corresponding to a rotation by $e^{\pi i/2}$, so one can envision how this might ``unroll" to a $\R^2$ example. However, this will still have degenerate projections near $t=1/2$.

\section{Minimizers of the MQV and $\mc{S}^{2,d}(X)$}\label{sec:param}

This section provides a proof for Theorem \ref{thm:1}. The proof exploits the vector-valued Fourier transform. Because we use closed curves, we have $\Phi[x]\in (L^2([0,1]))^2$ has the vector-valued Fourier representation
\begin{align}
\Phi[x](t)= \sum_{k\in\mathbb{Z}} \widehat{\Phi[x]}(k)e^{2\pi i k t}
\end{align}
where $\widehat{\Phi[x]}(k)\in \mathbb{C}^2$ and $\widehat{\Phi[x]}(-k)=\overline{\widehat{\Phi[x]}(-k)}$ since our curves are real-valued. 

For the next lemma, we define the Stiefel manifold $\text{St}(d,M)$ to be the set of all matrices $V\in \mathbb{C}^{d\times M}$ such that $VV^\ast=I_d$

\begin{lemma}\label{lem:1}
Suppose $0\leq \lambda_1 \leq \lambda_2 \leq \cdots \leq \lambda_M$, set $D=\text{\em diag}(\lambda_1,\lambda_2,\ldots, \lambda_M)\in \R^{M\times M}$. Then
\begin{align}
\sum_{k=1}^d \lambda_k = \min_{V\in \text{St}(d,M)} \text{\em trace}(V D V^\ast)
\end{align}
and $VDV^\ast= \sum_{k=1}^d \lambda_k$ if and only if there is a $Q\in \mc{U}(d)$ such that the rows of $Q^\ast V$ are an orthonormal collection of eigenvectors corresponding to the eigenvalues $\lambda_1,\ldots,\lambda_d$. 
\end{lemma}

\begin{proof}
The lower bound follows from the Courant-Fischer min-max theorem. To obtain the optimality characterization, note that the constraints are regular, and hence the Lagrange conditions become
\begin{align}
V D = \Lambda V\text{ and } VV^\ast=I_d
\end{align}
for some Hermitian matrix of Lagrange multipliers $\Lambda\in \R^{d\times d}$. Since $\Lambda$ is Hermitian, it is diagonalizable, and in particular there is a unitary matrix $Q\in \mc{U}(d)$ such that $\Lambda = Q\Theta Q^\ast$ where $\Theta$ is a diagonal matrix with non-increasing diagonal entries. Then
\begin{align}
Q^\ast V D = \Theta Q^\ast V.
\end{align}
Therefore the rows of $Q^TV$ form eigenvectors of $D$ with eigenvalues $\theta_k$. Moreover, $Q^\ast V(Q^\ast V)^\ast=I_d$ gives that these eigenvectors are orthonormal. Finally,
\begin{align}
\text{trace} (V DV^\ast)= \text{trace}(\Lambda VV^\ast) = \text{trace}(\Lambda)=\sum_{k=1}^d \theta_k
\end{align}
so eigenvalues $\{\theta_k\}_{k=1}^d$ must match with the eigenvalues $\{\lambda_k\}_{k=1}^d$ to ensure equality for the lower bound.
\end{proof}

This lemma allows us to precisely characterize the degrees of freedom of solutions to trace optimization problems subject to orthonormality constraints. 

\begin{lemma}\label{lem:mqv}
Suppose $X\in \R^{d\times N}$ has the singular value decomposition $X=U\Sigma V^T$ \footnote{We note that the singular values (i.e. the diagonal entries of $\Sigma$) are in the standard, non-increasing order.}. For $\Phi\in \mc{H}^{2\times d}$, define $\widetilde{\Phi}\in\mc{H}^{2\times d}$ by setting $\widetilde{\Phi}(t) = \Phi(t)U^T$, and set
\begin{align}
h_{j,k}= \left\Vert \frac{d\widetilde{\phi}_{j,k}}{dt}\right\Vert_{L^2}^2
\end{align}
for all $j\in[2]$ and $k\in[d]$. Then
\begin{align}
\text{\em MQV}(\Phi;X) = \sum_{k=1}^{d-1}&(\sigma_k^2-\sigma_{k+1}^2) \sum_{l=1}^k (h_{1,l} + h_{2,l}) + \sigma_d^2 \sum_{k=1}^d (h_{1,k} + h_{2,k})
\end{align}
\end{lemma}

\begin{proof}
For any $x\in\R^d$, linearity of integration and matrix multiplication yield
\begin{align}
\text{QV}(\Phi[x]) &= \int_0^1\left\Vert \frac{d\Phi[x]}{dt}(t) \right\Vert^2\:dt\\
&= \int_0^1\left[ \left( \frac{d\phi_{1,\cdot}}{dt}(t) x \right)^2 + \left( \frac{d\phi_{2,\cdot}}{dt}(t) x \right)^2 \right]\:dt\\
&= \int_0^1\left[ x^T \left(\frac{d\phi_{1,\cdot}}{dt}(t)\right)^T\frac{d\phi_{1,\cdot}}{dt}(t) x  + x^T \left(\frac{d\phi_{2,\cdot}}{dt}(t)\right)^T\frac{d\phi_{2,\cdot}}{dt}(t) x \right]\:dt\\
& = x^T \left(\int_0^1\left(\frac{d\phi_{1,\cdot}}{dt}(t)\right)^T\frac{d\phi_{1,\cdot}}{dt}(t)\:dt \right) x + x^T \left(\int_0^1\left(\frac{d\phi_{2,\cdot}}{dt}(t)\right)^T\frac{d\phi_{2,\cdot}}{dt}(t)\:dt \right) x
\end{align}
Define $G_{j,j^\prime}(\Phi)\in\R^{d\times d}$ by
\begin{align}
G_{j,j^\prime}(\Phi) = \int_0^1\left(\frac{d\phi_{j,\cdot}}{dt}(t)\right)^T\frac{d\phi_{j^\prime,\cdot}}{dt}(t)\:dt\text{ for } j,j^\prime\in[2].\label{defG}
\end{align}
Then we have 
\begin{align}
\text{QV}(\Phi[x]) & =x^T G_{1,1}(\Phi) x + x^T G_{2,2}(\Phi) x\\
& = \text{trace}\left[ \begin{pmatrix} xx^T & 0\\
0 & xx^T
\end{pmatrix} \begin{pmatrix}
G_{1,1}(\Phi) & 0 \\
0 & G_{2,2}(\Phi)
\end{pmatrix}\right].
\end{align}
From this expression, it follows that
\begin{align}
\text{MQV}(\Phi; x)  = \frac{1}{N}\text{trace} \left[\begin{pmatrix} XX^T & 0\\
0 & XX^T
\end{pmatrix} \begin{pmatrix}
G_{1,1}(\Phi) & 0 \\
0 & G_{2,2}(\Phi)
\end{pmatrix}\right]
\end{align}

Setting $\Lambda =\Sigma \Sigma^T$, note that $U^TG_{j,j}(\Phi)U=G_{j,j}(\widetilde{\Phi})$, and so invariance of the trace under conjugation by $U$ gives
\begin{align}
\text{MQV}(\Phi;X) &= \frac{1}{N} \text{trace} \left[ \begin{pmatrix}\Lambda & 0\\
0 & \Lambda
\end{pmatrix} \begin{pmatrix}
G_{1,1}(\widetilde{\Phi}) &0 \\
0 & G_{2,2}(\widetilde{\Phi})
\end{pmatrix}\right]\\
& = \sigma_1^2 h_{1,1} + \sigma_1^2 h_{2,1} + \cdots + \sigma_d^2 h_{1,d} + \sigma_d^2 h_{2,d}
\end{align}
where the last line follows because $\Lambda = \text{diag}(\sigma_1^2,\ldots, \sigma_d^2)$ and 
\begin{align}
\text{diag}(G_{j,j}(\widetilde{\Phi}))=\begin{pmatrix} h_{j,1} & h_{j,2} & \cdots & h_{j,d}\end{pmatrix}^T
\end{align}
for $j\in[2]$. The Abel summation formula yields
\begin{align}
\text{MQV}(\Phi;X) = &(\sigma_1^2-\sigma_2^2) (h_{1,1} + h_{2,1})\\
&+ (\sigma_2^2 - \sigma_3^2)(h_{1,1} + h_{2,1} + h_{1,2} + h_{2,2})\\
& + \ldots\\
& + (\sigma_{d-1}^2 - \sigma_d^2)\left( \sum_{k=1}^{d-1}(h_{1,k} + h_{2,k})\right) + \sigma_d^2 \sum_{k=1}^d (h_{1,k} + h_{2,k})
\end{align}
and the result is complete.
\end{proof}

The form of the $\text{MQV}$ derived in Lemma \ref{lem:mqv} allows us to also characterize the minimum value of the $\text{MQV}$. This minimum value is specified in Lemma \ref{lem:mqvBnd}.

\begin{lemma}\label{lem:mqvBnd}
Suppose $d$ and $N$ are natural numbers with $d\leq N$ and let $X\in\R^{d\times N}$ have a singular value decomposition $X=U\Sigma V^T$. For all $\Phi\in\mc{H}^{2\times d}$, satisfying
\begin{align}
\langle\phi_{j,k},\phi_{j^\prime, k^\prime}\rangle_{L^2} =\delta_{(j,k),(j^\prime,k^\prime)}\text{ for all }j,j^\prime\in[2],\:k,k^\prime\in[d]
\end{align}
and
\begin{align}
\langle\phi_{j,k}, {\bf 1}_{[0,1]}\rangle_{L^2}=0\text{ for all }j\in[2],\: k\in[d],
\end{align}
\begin{align}
\text{\em MQV}(\Phi;X) \geq 2\sum_{s=1}^{d-1}(\sigma_s^2-\sigma_{s+1}^2) \sum_{k=1}^s k^2 + 2\sigma_d^2 \sum_{k=1}^d k^2 \label{mqvkbnd}
\end{align}
Moreover, equality holds if and only if
\begin{align}
\sum_{k=1}^{s} (h_{1,k} + h_{2,k}) = 2 \sum_{k=1}^s k^2\label{mqveq}
\end{align}
for all $s$ such that $\sigma_s \not= \sigma_{s+1}$ or $s=d$, where $h_{j,k}$ is defined in the previous lemma.
\end{lemma}

\begin{proof} Without loss of generality, we assume $U=I$ so that $\widetilde{\Phi}=\Phi$ from the previous lemma. We define the block matrix
\begin{align}
G = \begin{pmatrix}
G_{1,1}(\Phi) & G_{1,2}(\Phi)\\
G_{2,1}(\Phi) & G_{2,2}(\Phi)
\end{pmatrix}
\end{align}
where $G_{j,j^\prime}(\Phi)$ is defined by equation (\ref{defG}), and we define the submatrices
\begin{align}
G^{(s)} = \begin{pmatrix}
G_{1,1}^{(s)}(\Phi) & G_{1,2}^{(s)}(\Phi)\\
G_{2,1}^{(s)}(\Phi) & G_{2,2}^{(s)}(\Phi)
\end{pmatrix}
\end{align}
where $G_{j,j^\prime}^{(s)}(\Phi)$ is the $s$ by $s$ submatrix in the upper left corner of $G_{j,j^\prime}(\Phi)$ for $j,j^\prime\in[2]$. 

Since $\phi_{j,k}\in \mc{H}$, $\phi_{j,k}$ admits a Fourier transform $\hat \phi_{j,k}\in\ell^2(\mathbb{Z})$. The condition $\int_0^1 \phi_{j,k}(t)\:dt=0$ is equivalent to $\hat\phi_{j,k}(0)=0$. Moreover, Parseval-Plancherel provides
\begin{align}
\left[G_{j,j^\prime}(\Phi)\right]_{k,k^\prime}=\left\langle  \frac{d\phi_{j,k}}{dt},\:\frac{d\phi_{j^\prime,k^\prime}}{dt}\right\rangle_{L^2} = \sum_{m\in\mathbb{Z}} m^2 \hat \phi_{j,k}(m) \overline{\hat \phi_{j^\prime,k^\prime}(m)}\label{infsum}
\end{align}
For all natural numbers $M$, we set
\begin{align}
\hat\Phi_{s,M} = \begin{pmatrix}
\hat\phi_{1,1}(-M) & \hat\phi_{1,1}(-M+1) & \cdots & \hat\phi_{1,1}(-1) & \hat\phi_{1,1}(1) & \cdots & \hat\phi_{1,1}(M-1) & \hat\phi_{1,1}(M)\\
\vdots & \vdots & \ddots & \vdots & \vdots & \ddots & \vdots & \vdots\\
\hat\phi_{1,s}(-M) & \hat\phi_{1,s}(-M+1) & \cdots & \hat\phi_{1,s}(-1) & \hat\phi_{1,s}(1) & \cdots & \hat\phi_{1,s}(M-1) & \hat\phi_{1,s}(M)\\
\hat\phi_{2,1}(-M) & \hat\phi_{2,1}(-M+1) & \cdots & \hat\phi_{2,1}(-1) & \hat\phi_{2,1}(1) & \cdots & \hat\phi_{2,1}(M-1) & \hat\phi_{2,1}(M)\\
\vdots & \vdots & \ddots & \vdots & \vdots & \ddots & \vdots & \vdots \\
\hat\phi_{2,s}(-M) & \hat\phi_{2,s}(-M+1) & \cdots & \hat\phi_{2,s}(-1) & \hat\phi_{2,s}(1) & \cdots & \hat\phi_{2,s}(M-1) & \hat\phi_{2,s}(M)\\
\end{pmatrix}
\end{align}
so that $\hat\Phi_{s,M} \in \mathbb{C}^{2s\times 2M}$. More compactly, we have that
\begin{align}
\hat\Phi_{s,M} = \begin{pmatrix}
\hat\Phi_{1,s,M}\\ \hat\Phi_{2,s,M}
\end{pmatrix}
\end{align}
with $\hat\Phi_{j,s,M}\in\mathbb{C}^{s\times 2M}$. Now, define $\mc{Q}_{M}\in\mathbb{R}^{2M\times 2M}$ by
\begin{align}
\mc{Q}_{M} = \text{diag}(M^2, (M-1)^2, \ldots, 1, 1,\ldots, (M-1)^2, M^2).
\end{align}
so that
\begin{align}
\hat\Phi_{s,M} \mc{Q}_{M} \hat\Phi_{s,M}^\ast = \begin{pmatrix}
\hat\Phi_{1,s,M}  \mc{Q}_{M} \hat\Phi_{1,s,M}^\ast & \hat\Phi_{1,s,M}  \mc{Q}_{M} \hat\Phi_{2,s,M}^\ast \\
\hat\Phi_{2,s,M}  \mc{Q}_{M} \hat\Phi_{1,s,M}^\ast & \hat\Phi_{2,s,M}  \mc{Q}_{M} \hat\Phi_{2,s,M}^\ast 
\end{pmatrix}.
\end{align}
Furthermore, 
\begin{align}
[\hat\Phi_{j,s,M}  \mc{Q}_{M} \hat\Phi_{j^\prime,s,M}^\ast]_{k,k^\prime} = \sum_{\vert m\vert \leq M} m^2 \hat \phi_{j,k}(m) \overline{\hat \phi_{j^\prime,k^\prime}(m)}.
\end{align}
These are partial sums of the convergent infinite series in Equation \ref{infsum}, so we may conclude that $\hat\Phi_{s,M} \mc{Q}_{M} \hat\Phi_{s,M}^\ast\to G^{(s)}$ as $M\to\infty$.

We now define
\begin{align}
\mc{Q}_{s,M} = \begin{pmatrix}
(s+1)^2 I_{M-s} & 0 & 0 \\
0 & \mc{Q}_s & 0\\
0 & 0 & (s+1)^2 I_{M-s}
\end{pmatrix}
\end{align}
for all $M > s$ where
\begin{align}
\mc{Q}_{s} = \text{diag}(s^2, (s-1)^2, \ldots, 1, 1,\ldots, (s-1)^2, s^2).
\end{align}
Now, observe that
\begin{align}
 [\hat \Phi_{j,s,M} \mc{Q}_{s,M} \hat\Phi_{j,s,M}^\ast]_{k,k^\prime} = \sum_{\vert m\vert \leq s}m^2 \hat \phi_{j,k}(m) \overline{\hat \phi_{j^\prime,k^\prime}(m)} + (s+1)^2 \sum_{s < \vert m \vert \leq M} \hat \phi_{j,k}(m) \overline{\hat \phi_{j^\prime,k^\prime}(m)}
\end{align}
This converges as $M\to\infty$ since the tail is a scalar multiple of the convergent series $\langle \phi_{j,k}, \phi_{j^\prime,k^\prime}\rangle_{L^2}=\sum_{ m \in\mathbb{Z}} \hat \phi_{j,k}(m) \overline{\hat \phi_{j^\prime,k^\prime}(m)}$ (again by Parseval-Plancherel since $\phi_{j,k}\in \mc{H}$). 
Since all of these series converge, we set
\begin{align}
H^{(s)}=\lim_{M\to\infty} \hat \Phi_{s,M} \mc{Q}_{s,M} \hat\Phi_{s,M}^\ast.
\end{align}
In particular 
\begin{align}
\lim_{M\to\infty} \sum_{\vert m\vert \leq s} m^2 \vert \hat\phi_{1,1}(m)\vert^2 + (s+1)^2\sum_{s<\vert m\vert \leq M} \vert \hat\phi_{1,1}(m)\vert^2=H^{(s)}_{1,1}\text{ for all }s\in[d],
\end{align}
and
\begin{align}
1 = \Vert \phi_{1,1}\Vert^2_{L^2}=\sum_{m\in\mathbb{Z}\setminus\{0\}}\vert \hat\phi_{1,1}(m)\vert^2
\end{align}
implies $\Vert H^{(s)}\Vert >0$ for all $s\in[d]$. Therefore, there is an $M_0$ such that $M\geq M_0$ implies 
\begin{align}
\Vert\hat \Phi_{s,M} \mc{Q}_{s,M} \hat\Phi_{s,M}^\ast - H^{(s)}\Vert\leq \Vert H^{(s)}\Vert.
\end{align} 
The reverse triangle inequality then establishes the bounds
\begin{align}
\Vert\hat \Phi_{s,M} \mc{Q}_{s,M} \hat\Phi_{s,M}^\ast\Vert\leq 2\Vert H^{(s)}\Vert\text{ for all }M\geq M_0.\label{normBnd}
\end{align}

Let $\varepsilon>0$, and set
\begin{align}
\varepsilon^\prime = \left(\frac{\varepsilon}{2\sqrt{2s}\Vert H^{(s)}(\Phi)\Vert+\varepsilon}\right) \wedge \frac{1}{2}
\end{align}
Observing that $\lim_{M\to\infty} \hat\Phi_{s,M}\hat\Phi_{s,M}^\ast = I_{2s}$, let $M^\prime$ be so large that $M\geq M^\prime$ implies
\begin{align}
\Vert \hat\Phi_{s,M}\hat\Phi_{s,M}^\ast - I_{2s}\Vert <\varepsilon^\prime.
\end{align}
This implies that the eigenvalues of $\hat\Phi_{s,M}\hat\Phi_{s,M}^\ast$ are bounded below by $\frac{1}{2}\leq 1-\varepsilon^\prime$. Thus, $\hat\Phi_{s,M}\hat\Phi_{s,M}^\ast$ is invertible for all $M\geq M^\prime$. Moreover, the eigenvalues of $\hat\Phi_{s,M}\hat\Phi_{s,M}^\ast$ satisfy
\begin{align}
\vert 1-\nu_k^{-1}\vert \leq \frac{\varepsilon^\prime}{1-\varepsilon^\prime}\text{ for all } M\geq M^\prime
\end{align}
Taking the square root sum of squares we get
\begin{align}
\left\Vert I_{2s} - (\hat\Phi_{s,M}\hat\Phi_{s,M}^\ast)^{-1}\right\Vert \leq \sqrt{2s} \frac{\varepsilon^\prime}{(1-\varepsilon^\prime)}\leq \frac{\varepsilon}{2\Vert H^{(s)}\Vert}.\label{ubinv}
\end{align}
Also note that $V_{s,M} = (\Phi_{s,M} \Phi_{s,M}^\ast)^{-1/2} \Phi_{s,M}$ has orthonormal rows for $M\geq M^\prime$.

Then (with $h_{j,k}$ defined in Lemma \ref{lem:mqv}),
\begin{align}
&\sum_{k=1}^s(h_{1,k}+h_{2,k})\\
 & = \sum_{k=1}^s \sum_{m\in\mathbb{Z}} m^2 ( \vert\hat\phi_{1,k}(m)\vert^2 + \vert\hat\phi_{2,k}(m)\vert^2)\\
&\geq \sum_{k=1}^s \left[\sum_{\vert m\vert\leq s} m^2 (\vert\hat\phi_{1,k}(m)\vert^2 + \vert\hat\phi_{2,k}(m)\vert^2) + (s+1)^2 \sum_{s < \vert m \vert }(\vert\hat\phi_{1,k}(m)\vert^2 + \vert\hat\phi_{2,k}(m)\vert^2)\right]\label{lb2}\\
& \geq \sum_{k=1}^s \left[\sum_{\vert m\vert\leq s} m^2 (\vert\hat\phi_{1,k}(m)\vert^2 + \vert\hat\phi_{2,k}(m)\vert^2) + (s+1)^2 \sum_{\substack{
s < \vert m \vert \\
\vert m \vert \leq M}
}(\vert\hat\phi_{1,k}(m)\vert^2 + \vert\hat\phi_{2,k}(m)\vert^2)\right]\label{lb1}\\
&=\text{trace}\left(\hat\Phi_{s,M} \mc{Q}_{s,M} \hat\Phi_{s,M}^\ast\right) - \text{trace}\left(V_{s,M} \mc{Q}_{s,M} V_{s,M}^\ast\right) + \text{trace}\left(V_{s,M} \mc{Q}_{s,M} V_{s,M}^\ast\right)\\
& = \text{trace}\left(V_{s,M} \mc{Q}_{s,M} V_{s,M}^\ast\right) + \text{trace}\left(\left(I_{2s} - (\hat\Phi_{s,M}\hat\Phi_{s,M}^\ast)^{-1}\right)\hat\Phi_{s,M} \mc{Q}_{s,M} \hat\Phi_{s,M}^\ast\right)
\end{align}
Recall that
\begin{align}
\left\Vert \hat\Phi_{s,M} \mc{Q}_{s,M} \hat\Phi_{s,M}^\ast\right\Vert \leq 2\Vert H^{(s)}\Vert
\end{align}
for all $M\geq M_0$. Cauchy-Schwarz on the Hilbert-Schmidt inner product now yields
\begin{align}
\sum_{k=1}^s(h_{1,k}+h_{2,k}) & \geq \text{trace}\left(V_{s,M} \mc{Q}_{s,M} V_{s,M}^\ast\right) - \left\Vert I_{2s} - (\hat\Phi_{s,M}\hat\Phi_{s,M}^\ast)^{-1}\right\Vert \left\Vert \hat\Phi_{s,M} \mc{Q}_{s,M} \hat\Phi_{s,M}^\ast\right\Vert\\
& \geq \text{trace}\left(V_{s,M} \mc{Q}_{s,M} V_{s,M}^\ast\right) - \varepsilon,
\end{align}
where this last bound follows from the inequalities (\ref{normBnd}) and  (\ref{ubinv}) as long as $M > \max(s,M_0, M^\prime)$.

By applying Lemma \ref{lem:1} to $\mc{Q}_{s,M}$ and $V_{s,M}$,
\begin{align}
 \text{trace}\left(V_{s,M} \mc{Q}_{s,M} V_{s,M}^\ast\right) \geq \sum_{k=1}^s 2k^2,
\end{align}
and we conclude that
\begin{align}
&\sum_{k=1}^s \left[\sum_{\vert m\vert\leq s} m^2 (\vert\hat\phi_{1,k}(m)\vert^2 + \vert\hat\phi_{2,k}(m)\vert^2) + (s+1)^2 \sum_{s < \vert m \vert }(\vert\hat\phi_{1,k}(m)\vert^2 + \vert\hat\phi_{2,k}(m)\vert^2)\right]\\
 &\geq\left( \sum_{k=1}^s 2k^2\right) - \varepsilon.
\end{align}
Since $\varepsilon>0$ was arbitrary, the approximation property of inequalities yields
\begin{align}
\sum_{k=1}^s \left[\sum_{\vert m\vert\leq s} m^2 (\vert\hat\phi_{1,k}(m)\vert^2 + \vert\hat\phi_{2,k}(m)\vert^2) + (s+1)^2 \sum_{s < \vert m \vert }(\vert\hat\phi_{1,k}(m)\vert^2 + \vert\hat\phi_{2,k}(m)\vert^2)\right] \geq \sum_{k=1}^s 2k^2.\label{lb3}
\end{align}
Then the inequality (\ref{lb2}) gives
\begin{align}
\sum_{k=1}^s(h_{1,k}+h_{2,k}) & \geq \sum_{k=1}^s 2k^2.\label{partmqvbnd}
\end{align}
for all $s=1,\ldots, d$. The non-increasing order of $\sigma_s$ gives $\sigma_s-\sigma_{s+1}\geq 0$ for all $s$, and hence the inequalities in Equation \ref{partmqvbnd} imply the bound in Equation \ref{mqvkbnd} by Lemma \ref{lem:mqv}.

Finally, we note that equality in Equation \ref{mqvkbnd} holds trivially if Equation \ref{mqveq} holds for all $s$ such that $\sigma_s\not=\sigma_{s+1}$ or $s=d$ by Lemma \ref{lem:mqv}. On the other hand, if equality holds in Equation \ref{mqvkbnd}, then Lemma \ref{lem:mqv} implies
\begin{align}
\sum_{s=1}^{d-1} (\sigma_s^2-\sigma_{s+1}^2) \left(\sum_{k=1}^s(h_{1,k}+h_{2,k})  - 2\sum_{k=1}^sk^2\right) + 2\sigma_d^2\left ( \sum_{k=1}^d(h_{1,k}+h_{2,k}) - 2\sum_{k=1}^d k^2\right)=0.\label{exactmqv}
\end{align}
Since the $\sigma_k$'s form a non-increasing sequence, $\sigma_s-\sigma_{s+1}\geq 0$ and 
\begin{align}
\sum_{k=1}^s(h_{1,k}+h_{2,k})  - 2\sum_{k=1}^sk^2\geq 0
\end{align}
for all $s\in[d]$ by Equation \ref{partmqvbnd}. Consequently, Equation \ref{exactmqv} requires that Equation \ref{mqveq} holds for all $s$ such that $\sigma_s\not=\sigma_{s+1}$ or $s=d$. This concludes the proof.
\end{proof}

With Lemma \ref{lem:mqvBnd}, we are now able to prove Theorem \ref{thm:1}.\\

\begin{proof}[Proof of Theorem \ref{thm:1}]
Without loss of generality, the singular value decomposition $X=U\Sigma V^T$ satisfies $U=I$, so $\widetilde{\Phi}$ is just $\Phi$. 

We first note that, for all $s$ such that $\sigma_s\not = \sigma_{s+1}$ or $s=d$, the system of equalities from Lemma \ref{lem:mqvBnd}
\begin{align}
\sum_{k=1}^s (h_{1,k} + h_{2,k}) = 2\sum_{k=1}^s k^2
\end{align}
 is equivalent to the system of equalities
\begin{align}
\sum_{k=k_{q,\min}}^{k_{q,\max}} (h_{1,k} + h_{2,k}) = 2\sum_{k=k_{q,\min}}^{k_{q,\max}}k^2\text{ for all }q\in[p].
\end{align}

First suppose that $\Phi_{[k_{q,\min},k_{q,\max}]}\in\Pi_{[k_{q,\min},k_{q,\max}]}$ for all $q\in[p]$. Then there is a $Q_q\in\mc{O}(2(k_{q,\max}-k_{q,\min}+1))$ such that
\begin{align}
Q_q^T\cdot \Phi_{[k_{q,\min},k_{q,\max}]} = \Psi_{k_{q,\min},k_{q,\max}}
\end{align}
because $Q_q$ preserves norms, we have that
\begin{align}
\sum_{k=k_{q,\min}}^{k_{q,\max}}(h_{1,k} + h_{2,k}) = \left\Vert \frac{d\Phi_{[k_{q,\min},k_{q,\max}]}}{dt} \right\Vert_{L^2}^2 = \left\Vert \frac{d\Psi_{k_{q,\min},k_{q,\max}}}{dt}\right\Vert_{L^2}^2 = 2\sum_{k=k_{q,\min}}^{k_{q,\max}}k^2.
\end{align}
Thus, $\Phi$ is a minimizer by Lemma \ref{lem:mqvBnd}.

To prove the converse, we first note that if $\Phi$ has any component $\phi_{j,k}$ such that $\hat\phi_{j,k}$ does not have finite support, then the inequality (\ref{lb2}) is strict for $s=d$, so Lemma \ref{lem:mqvBnd} indicates that $\Phi$ cannot be a minimizer. 

Thus, if $\Phi$ is a minimizer, then all of the sequences $\hat\phi_{j,k}$ have finite support, and hence there exists a natural number $M$ such that $\hat\Phi_{s,M}$ from Lemma \ref{lem:mqvBnd} satisfies 
\begin{align}
\text{trace}\left( \hat\Phi_{s,M} \mc{Q}_M \hat\Phi_{s,M}^\ast\right) = 2\sum_{k=1}^sk^2,
\end{align}
and therefore Lemma \ref{lem:1} gives a $Q\in \mc{U}(2M)$ such that $Q^\ast \hat\Phi_{s,M}$ has orthonormal rows of eigenvectors of $\mc{Q}_M$ corresponding to the lowest $2s$ eigenvalues of $\mc{Q}_M$. Consequently, we must have that $\Phi_{[1,s]}\in \Pi_{[1,s]}$ for $s=d$ and all $s$ with $\sigma_s\not=\sigma_{s+1}$. Then $\Phi_{[k_{1,\min},k_{1,\max}]}\in \Pi_{[k_{1,\min},k_{1,\max}]}$ and an induction argument establishes the result.

\end{proof}

\subsection{Discussion}

We note that
\begin{align}
 \Vert \Phi[x](t)\Vert \leq \sum_{k\in\mathbb{Z}} \left\Vert\widehat{\Phi[x]}(k)\right\Vert &= \sum_{k\in\mathbb{Z}} \frac{1}{k}\left\Vert k\widehat{\Phi[x]}(k)\right\Vert\\& \leq \sqrt{\sum_{k\in\mathbb{Z}}\frac{1}{k^2}}\sqrt{\sum_{k\in\mathbb{Z}}\left\Vert k\widehat{\Phi[x]}(k)\right\Vert^2}\\
 & = C \left\Vert \frac{d\Phi[x]}{dt}\right\Vert_{L^2},
\end{align}
for all $t\in[0,1]$, so $\Vert \Phi[x]\Vert_{L^\infty}\leq C \left\Vert \frac{d\Phi[x]}{dt}\right\Vert_{L^2}$ for all $x\in\R^d$. Thus, we have that minimization of the mean quadratic variation also squeezes the gap present in the bound
\begin{align}
\Vert \Phi[x_j]\Vert_{L^1} \leq \Vert \Phi[x_j]\Vert_{L^2}\leq \Vert \Phi[x_j]\Vert_{L^\infty}.
\end{align}

The main theorem may be imitated for curves embedded in arbitrary dimension. However, different principal components are identified with similar frequencies for 1D curves, and the assignment of the principal components in the $\R^3$ is staggered. The plots of such curves are still informative, but lack the symmetry enjoyed by the $\R^2$ embeddings. In particular, the $\R^3$ case does not generally enjoy the property that projections onto 1D subspaces of $\R^2$ result in the same 1D MQV. 

This result may also be imitated for other quadratic forms, but now the result will identify PCA components to eigenfunctions of some other operator. For example, the Legendre basis arises from minimizing the discrete Legendre equations, and the Hermite functions come from the Schrodinger-Laplacian.  This may have some relevance with respect to the work \cite{embrechts1991variations}, but we do not explore this possibility here.

Theorem \ref{thm:1} also suggests that the optimal embeddings are related to partial Laurent series with a $0$ constant component under the identification of $\R^2$ with $\mathbb{C}$.

\section{Quadratic phase shifts and the asymptotic tour property}\label{sec:tour}

This section provides the proof of Theorem \ref{thm:2}. We begin with a lemma reminiscent of Theorem 2.7 of \cite{goyal2001quantized}: 
\begin{lemma}\label{lem:perturb}
Suppose $Z=\{z_k\}_{k=1}^d\subset\mathbb{C}$ and $w\in\mathbb{C}$ satisfy
\begin{align}
\sum_{k=1}^d \vert z_k\vert^2 = 2\text{ and }\sum_{k=1}^d z_k^2 = \omega^2,
\end{align}
and set $a_k=\text{Re}(z_k)$ and $b_k=\text{Im}(z_k)$ for $k\in[d]$. Then the singular values of the matrix
\begin{align}
Z=\begin{pmatrix}
a_1 & \cdots &a_d\\
b_1 & \cdots &b_d
\end{pmatrix}
\end{align}
are $\sqrt{1\pm \frac{\vert \omega\vert^2}{2}}$.
\end{lemma}
\begin{proof}
We note that 
\begin{align}
A = ZZ^T = \begin{pmatrix}
\sum_{k=1}^d a_k^2 & \sum_{k=1}^d a_k b_k\\
\sum_{k=1}^d a_k b_k & \sum_{k=1}^d b_k^2
\end{pmatrix}
\end{align}
has trace
\begin{align}
\sum_{k=1}^d a_k^2+ \sum_{i=1}^d b_k^2  = \sum_{i=1}^d \vert z_k\vert^2 = 2.
\end{align}
Let $\omega = v + i w$, set
\begin{align}
\widehat{Z} = \begin{pmatrix}
a_1 & \cdots &a_d & -w\\
b_1 & \cdots &b_d & v
\end{pmatrix}
\end{align}
and observe that
\begin{align}
2\left(\sum_{k=1}^d a_k b_k -wv\right) = \text{Im}\left(\sum_{k=1}^d z_k^2 - \omega^2\right) = 0,
\end{align}
and
\begin{align}
\sum_{k=1}^d a_k^2 + w^2 - \sum_{k=1}^d b_k^2 -v^2 = \text{Re}\left( \sum_{k=1}^d z_k^2 - \omega^2\right) = 0
\end{align}
so $\widehat{Z}\widehat{Z}^T$ is a diagonal matrix with a constant diagonal given by $c = \sum_{k=1}^d a_k^2 + w^2$. If $\omega=0$, then $ZZ^T=\widehat{Z}\widehat{Z}^T$ and $\text{trace}(ZZ^T)=2$ gives $ZZ^T=I_2$ and the result follows.

On the other hand, if $\omega\not=0$, then
\begin{align}
ZZ^T = \widehat{Z}\widehat{Z}^T - \begin{pmatrix} -w\\ v\end{pmatrix}\begin{pmatrix} -w\\ v\end{pmatrix}^T= cI_2 - \begin{pmatrix} -w\\ v\end{pmatrix}\begin{pmatrix} -w\\ v\end{pmatrix}^T,
\end{align}
the eigenvectors of $ZZ^T$ are
\begin{align}
\begin{pmatrix} v\\ w\end{pmatrix}, \begin{pmatrix} -w\\ v\end{pmatrix}
\end{align}
with eigenvalues $c$ and $c-\vert\omega\vert^2$. Since $\text{trace}(ZZ^T)=2$, $2 = 2c - \vert\omega\vert^2$, and we conclude that $c = 1 + \frac{\vert\omega\vert^2}{2}$. This verifies that the eigenvalues of $ZZ^T$ are $1 \pm \frac{\vert\omega\vert^2}{2}$, and expression for the singular values of $Z$ follow.
\end{proof}

\begin{proof}[Proof of Theorem \ref{thm:2}]
We first discuss the organization of this proof. We first identify the columns of $\Phi(t)$ with complex numbers $z_k(t)$. Then, by Lemma \ref{lem:perturb}, a bound on the eigenvalues of $\Phi(t)$ follows from a bound on $\vert \sum_{k=1}^d z_k(t)^2\vert$. Equation 2.4 of \cite{paris2014asymptotic} gives an exact expression for this sum, and we provide bounds on all the terms in this expression using elementary techniques to complete the proof.

Let $t\in[0,1]$ and set $z_k=(\phi_{1,k}(t)+i \phi_{2,k}(t))/\sqrt{2}$ for $k\in[d]$.  Since 
\begin{align}
\begin{pmatrix}
\phi_{1,k}(t)\\
\phi_{2,k}(t)
\end{pmatrix} = \begin{pmatrix}
\cos(\pi k^2/2d) & -\sin(\pi k^2/2d)\\
\sin(\pi k^2/2d) & \cos(\pi k^2/2d)
\end{pmatrix}\begin{pmatrix}
\sqrt{2}\cos(2\pi k t)\\
\sqrt{2}\sin(2\pi k t)
\end{pmatrix},
\end{align}
it follows that $\vert z_k\vert^2 = 1$ for all $k\in[d]$, so $\sum_{k=1}^d\vert z_k\vert^2=d$. Moreover, we see that
\begin{align}
z_k = e^{2\pi i \frac{k^2}{4d}}e^{2\pi i kt}.
\end{align}
Then
\begin{align}
\sum_{k=1}^d z_k^2 = \sum_{k=1}^d e^{\pi i \frac{k^2}{d}}e^{2\pi i k(2t)}.
\end{align}
noting that
\begin{align}
e^{2\pi i k(2t)} = e^{2\pi i k(2t - 1)} = e^{2\pi i k (2t-2)}
\end{align}
for all $k\in[d]$, this sum has the form
\begin{align}
S_N(x,\theta) = \sum_{k=1}^N f(k) 
\end{align}
where $f(k) = e^{\pi i x k^2}e^{2\pi i k\theta}$, $N=d$, $x = 1/d \in (0,1)$, and $\theta\in[-1/2, 1/2]$. In the notation of \cite{paris2014asymptotic}, we let $\text{erfc}$ denote the complementary error function defined by the line integral
\begin{align}
\text{erfc}(z) = 1 - \text{erf}(z) = \frac{2}{\sqrt{\pi}}\int_{\vert z\vert}^\infty e^{-w^2}\:dw,
\end{align}
and set 
\begin{align}
E(\tau) = e^{-\pi i \tau^2/x} \text{erfc}\left(\omega \tau \sqrt{\pi/x}\right) = e^{-\pi i d \tau^2} \frac{2}{\sqrt{\pi}}\int_{\tau\sqrt{\pi/x}}^\infty e^{-i s^2}\:ds,
\end{align}
where $\omega = e^{-\pi i/4}$. Equation 2.4 of \cite{paris2014asymptotic} provides the expression
\begin{align}
S_N(x,\theta) = \frac{1}{2}(f(N)-1) + J_N + e^{\pi i/4} (I_N - I_0)
\end{align}
where
\begin{align}
J_N = \frac{e^{\pi i/4}}{2\sqrt{x}}\left(E(\theta) - f(N)E(\xi) \right),
\end{align}
for $\xi = Nx+\theta$ and 
\begin{align}
I_j = \frac{f(j)}{2\sqrt{x}}\sum_{k=1}^\infty \left( E(k-jx-\theta) - E(k+jx+\theta)\right)
\end{align}
for $j=0,N$.

We have that $\left\vert \frac{1}{2}(f(N)-1)\right\vert \leq 1$. Then using the substitution $s = u\sqrt{\pi/2}$ gives
\begin{align}
\left\vert \int_{\tau\sqrt{\pi/x}}^\infty e^{i s^2}\:ds\right\vert &= \left\vert \int_{\tau\sqrt{2/x}}^{\infty} \cos(\pi u^2/2)\:du + i  \int_{\tau\sqrt{2/x}}^{\infty} \sin(\pi u^2/2)\:du\right\vert \leq 2
\end{align}
and
\begin{align}
\vert J_N\vert \leq \frac{2}{\sqrt{x}}
\end{align}
It remains to bound the expressions
\begin{align}
I_j = \frac{f(j)}{2\sqrt{x}}\sum_{k=1}^\infty \left( E(k-jx-\theta) - E(k+jx+\theta)\right)
\end{align}
for $j=0,N$. We only need to consider bound for the sums
\begin{align}
\sum_{k=1}^\infty E(k-\theta) - E(k+\theta)
\end{align}
when $j=0$ and
\begin{align}
\sum_{k=1}^\infty E(k-1-\theta) - E(k+1+\theta)
\end{align}
when $j=N$ (and so $jx= Nx= d(1/d)=1$). We note that $f(0)=1$.

Using the asymptotic expansion in Equation 2.2 of \cite{paris2014asymptotic}\footnote{This cites \cite{olver2010nist} and error bound derivations are from \cite{olver1997asymptotics}.}, with $n=1$, we have that
\begin{align}
E(t) = \frac{1}{\sqrt{\pi}}\frac{\Gamma(0+\frac{1}{2})}{\Gamma(\frac{1}{2})}\left( \frac{ix}{\pi t^2}\right)^{0+\frac{1}{2}}+ T_1(t)= \frac{\sqrt{x}}{\pi} \frac{e^{\pi i/4}}{t} + T_1(t)
\end{align}
where $\vert T_1(t)\vert \leq \frac{\Gamma(1+\frac{1}{2})}{\pi}\left(\frac{x}{\pi t^2} \right)^{1+\frac{1}{2}} = \frac{x^{3/2}}{\pi^2}\frac{1}{t^3}$ and $t>0$. Therefore, for $k\geq 1$,
\begin{align}
E(k-\theta) - E(k+\theta) = \frac{\sqrt{x}}{\pi} e^{\pi i/4}(k-\theta)^{-1} + T_1(k-\theta) -  \frac{\sqrt{x}}{\pi} e^{\pi i/4}(k+\theta)^{-1} - T_1(k+\theta)
\end{align}
which gives
\begin{align}
E(k-\theta) - E(k+\theta) = \frac{\sqrt{x}}{\pi}e^{\pi i/4}\frac{2\theta}{k^2-\theta^2} + T_1(k-\theta) - T_1(k+\theta).
\end{align}
A bound for the remainder term is given by
\begin{align}
\vert T_1(k-\theta) - T_1(k+\theta)\vert \leq \frac{x^{3/2}}{2\pi^2} \left(\frac{1}{(k-\theta)^3} + \frac{1}{(k+\theta)^3}\right).
\end{align}
We have
\begin{align}
\sum_{k=1}^\infty \frac{1}{k^2-\theta^2} \leq \sum_{k=1}^\infty \frac{1}{k^2-\frac{1}{4}} = \sum_{k=1}^\infty \frac{1}{k^2-\frac{1}{4}} = 2
\end{align}
using the inductive formula $\sum_{k=1}^m \frac{1}{k^2-\frac{1}{4}} = \frac{4m}{2m+1}$. This provides the bound 
\begin{align}
\left \vert \sum_{k=1}^\infty \frac{\sqrt{x}}{\pi}e^{\pi i/4}\frac{2\theta}{k^2-\theta^2} \right\vert \leq \frac{2}{\pi}\sqrt{x}.
\end{align}
On the other hand, we have the bound
\begin{align}
\sum_{k=1}^\infty \frac{1}{(k-\theta)^3} + \frac{1}{(k+\theta)^3} &\leq 2 \sum_{k=1}^\infty \frac{1}{(k-\frac{1}{2})^3}
\end{align}
and
\begin{align}
\sum_{k=1}^\infty \frac{1}{(k-\frac{1}{2})^3} &= 8 + \frac{8}{27} + \sum_{k=3}^\infty \frac{1}{(k-\frac{1}{2})^3}\\
&< 8 + \frac{8}{27} + \int_2^\infty \frac{1}{(q - \frac{1}{2})^3}\:dq\\
&< 9.
\end{align}
This yields the bound
\begin{align}
\sum_{k=1}^\infty\vert T_1(k-\theta) - T_2(k+\theta)\vert \leq 9\frac{x^{3/2}}{\pi^2}.
\end{align}

A bound on $I_N$ and the second series follows in a similar fashion, but the asymptotic formula for $E(1-1-\theta)=E(-\theta)$ may not hold, so we use the estimate $\vert E(-\theta) - E(2+\theta) \vert \leq 4$ to avoid the first term. For $k\geq 2$, the terms satisfy
\begin{align}
E(k-1-\theta) - E(k+1+\theta) = \frac{\sqrt{x}}{\pi}e^{\pi i/4}\frac{2\theta}{k^2-(1+\theta)^2} + T_1(k-1-\theta) - T_1(k+1+\theta).
\end{align}
with a bound on the remainder term given by
\begin{align}
\vert T_1(k-1-\theta) - T_1(k+1+\theta)\vert \leq \frac{x^{3/2}}{2\pi^2} \left(\frac{1}{(k-1-\theta)^3} + \frac{1}{(k+1+\theta)^3}\right).
\end{align}
The sum of absolute values of the first terms for $k\geq 2$ is then bounded by $\frac{46}{45\pi}\sqrt{x}$ and the sum of absolute values of the remainder terms is bounded by $9\frac{x^{3/2}}{\pi^2}$. Bringing all these bounds together, we get the bound
\begin{align}
\vert I_0- I_N\vert \leq \frac{1}{2\sqrt{x}} \left(\frac{2}{\pi}\sqrt{x} + 9 \frac{x^{3/2}}{\pi^2} + 4 + \frac{46}{45\pi}\sqrt{x} + 9 \frac{x^{3/2}}{\pi^2}  \right)\leq \frac{1}{2} + x + \frac{2}{\sqrt{x}}.
\end{align}
The final bound is then
\begin{align}
\vert S_N(x,\theta)\vert &\leq \left \vert\frac{1}{2}(f(N)-1)\right\vert + \vert J_N\vert +\vert I_0-I_N\vert\\
& \leq 1 + \frac{2}{\sqrt{x}} + \frac{1}{2} + x + \frac{2}{\sqrt{x}}\\
& = 4\sqrt{d} + \frac{3}{2} + \frac{1}{d}
\end{align}
after the substitution $x=1/d$. Multiplying this bound by $2/d$ corresponds to multiplying all $z_i$ by the constant $\sqrt{\frac{2}{d}}$, which allows us invoke Lemma \ref{lem:perturb} to observe that $\sqrt{\frac{1}{d}}\Phi(t)$ has singular values between 
\begin{align}
\sqrt{1 \pm \left(\frac{4}{\sqrt{d}}+\frac{3}{2d}+\frac{1}{d^2}\right)}.
\end{align}
Since $t$ was arbitrary, we see that this bound holds for all $t$. 
\end{proof}

We note that this asymptotic rate for a uniform bound is optimal. This is because the expectation of the square magnitude satisfies
\begin{align}
\int_0^1\left\vert \sum_{k=1}^d \alpha_k e^{2\pi i k t}\right\vert^2\:dt = \sum_{k=1}^d \vert \alpha_k\vert^2,
\end{align}
so the maximum value of the magnitude is at least $\sqrt{d}$ given $\vert\alpha_k\vert=1$ for all $k\in[d]$.

\subsection{Discussion of quadratic Gauss sums}

The work from \cite{hardy1914some} derives the basic bound for our expression, but \cite{paris2014asymptotic} has a form amenable to expressing the bound explicitly for finite dimension. Lehmer \cite{lehmer1976incomplete} considers incomplete Gauss sums which are related to our approach and has the right form for asymptotics. This paper also has a geometric interpretation of the cancellations in the sum, and illustrates how cumulative generalized Gaussian sums form interesting spiral structures related to Euler spirals.

An elementary derivation of the quadratic Gauss sum is given in \cite{murty2017evaluation}, and it also details the resulting quadratic reciprocity result. The paper \cite{oskolkov1991functional} considers sums of the form required in this paper (generalized Gaussian sums), but the sums are infinite. The textbooks \cite{berndt1998gauss, iwaniec2004analytic} derive expressions for generalized Gaussian sums, but the range of the coefficients is outside of our interest.

The work \cite{bourgain2006gauss} establishes bounds for Gauss sums over additive characters of finite fields, and \cite{demirci2013value} consider limit distributions of Gauss sums. Coutsias and Kazarinoff \cite{coutsias1998approximate} considers a precise bound for the ``functional" approximation formula of Gauss sums. This applies for the small $1/d$ in our case, but works in the continuous setting and excludes the linear terms we need.

\section{Inducing 3D curves from Bishop frames}\label{sec:filament}

To convert the ``braided" 3D Andrews plots to ``bushy" filaments plots, we consider a non-linear map $\Gamma: \left( C^1([0,1])\right)^2\to \left( C^1([0,1])\right)^3$ such that $\Gamma[\phi]$ satisfies
\begin{enumerate}
\item $\Gamma[\phi](0)=0$
\item $\Gamma[\phi]^\prime$ consists of the first row of the solution to the matrix differential equation
\begin{align}
\begin{pmatrix}
d{\bf T}(t)\\
d{\bf N}_1(t)\\
d{\bf N}_2(t)
\end{pmatrix} = \begin{pmatrix}
0 & \phi_1(t) & \phi_2(t)\\
-\phi_1(t) & 0 & 0\\
-\phi_2(t) & 0 & 0
\end{pmatrix} \begin{pmatrix}
{\bf T}(t)\\
{\bf N}_1(t)\\
{\bf N}_2(t)
\end{pmatrix}
\end{align}
subject to the initial condition that 
\begin{align}
\begin{pmatrix}
{\bf T}(0)\\
{\bf N}_1(0)\\
{\bf N}_2(0)
\end{pmatrix} = \begin{pmatrix}
1 & 0 & 0\\
0 & 1 & 0\\
0 & 0 & 1
\end{pmatrix}.
\end{align}
Here $\phi_1$ and $\phi_2$ are the component function of $\phi:[0,1]\to \R^2$. 
\end{enumerate}

If unique Bishop frames always exist, then we can be certain that the map $\Gamma$ is well-defined. The following proposition provides the uniqueness and existence result. The proof is a standard exercise, but we include a sketch here for completeness.

\begin{proposition}
Consider a matrix-valued function $A:[0,1]\to\R^{3\times 3}$ such that
\begin{align}
L = \sup_{t\in[0,1]} \Vert A(t)\Vert
\end{align}
is finite. Then there exists a unique, differentiable $U:[0,1]\to \R^{3\times 3}$ such that
\begin{align}
U^\prime(t) = A(t) U(t)\text{ for all }t\in[0,1], \text{ and }U(0) = I.
\end{align}
In particular, unique solutions exist when $A$ has the form
\begin{align}
A(t) = \begin{pmatrix}
0 & \phi_1(t) & \phi_2(t)\\
-\phi_1(t) & 0 & 0\\
-\phi_2(t) & 0 & 0
\end{pmatrix}
\end{align}
for some $\phi_1,\phi_2\in C^1([0,1])$.

\end{proposition}

\begin{proof}
Consider the function $f(t, U(t)) = A(t) U(t)$.  We then set
\begin{align}
L = \sup_{t\in[0,1]} \Vert A(t)\Vert
\end{align}
where $\Vert \cdot \Vert$ is the spectral norm, and define the norm
\begin{align}
\Vert U\Vert_L = \sup_{t\in[0,1]}e^{-2L t}\Vert U(t)\Vert
\end{align}
for any $U:[0,1]\to \R^{3\times 3}$. Now, define the Picard operator
\begin{align}
\mathcal{P}[U](t) = I + \int_0^t f(s, U(s))\:ds.
\end{align}
This operator is contractive:
\begin{align}
e^{-2Lt}\Vert \mathcal{P}[U](t) - \mathcal{P}[V](t)\Vert &\leq \int_0^t e^{-2L(t-s)} e^{2Ls} L \Vert U(s)-V(s)\Vert\:ds\\
&\leq L \Vert U - V\Vert_L \int_0^t e^{-2L(t-s)}\:ds\\
&\leq \frac{1}{2} \Vert U-V\Vert_L.
\end{align}
Thus, $\Vert \mathcal{P}[U] - \mathcal{P}[V]\Vert_L\leq \frac{1}{2} \Vert U-V\Vert_L$ allows us to use the Banach fixed point theorem to obtain a unique fixed point of the Picard iteration starting at $U^{(0)}(t) = I$. 

Now, if 
\begin{align}
A(t) = \begin{pmatrix}
0 & \phi_1(t) & \phi_2(t)\\
-\phi_1(t) & 0 & 0\\
-\phi_2(t) & 0 & 0
\end{pmatrix}
\end{align}
for some $\phi_1,\phi_2\in C^1([0,1])$, then
\begin{align}
\Vert A(t) \Vert \leq \sqrt{\phi_1(t)^2 + \phi_2(t)^2}
\end{align}
is bounded over $t\in[0,1]$ by the extreme value theorem. Consequently, there is a unique solution to the initial value problem $U^\prime(t) = A(t) U(t)$ with $U(0)=I$. 
\end{proof}

\subsection{Constructing filaments plots for data} With $\Gamma$ defined, we now describe the process for constructing filament plots given a data matrix $X\in \R^{d\times N}$.

\begin{enumerate}
\item Let $\Phi\in \mc{L}(\R^d, \mc{H}^2)$ denote the map obtained from $X$ specified in Theorem \ref{thm:2}.
\item The filament plot for a data point $x$ is then $\Gamma[\Phi[x]]$. 
\end{enumerate}

In particular, we know that $\Gamma$ is defined on the range of $\Phi$ from Theorem \ref{thm:2} since the component functions of $\Phi[x]$ are always real analytic, so on the compact interval $[0,1]$ the functions are all Lipschitz.

The final issue in the construction of the filament plot involves the actual computation of $\Gamma[\phi]$. Production of closed form solutions for 3D moving frames is notoriously difficult, but numerical methods are available. For the examples in this paper, the filaments plots are produced by performing a third-order ``Lie" Runge-Kutta method \cite{crouch1993numerical, iserles2000lie} to numerically construct the Bishop frame. Rotations of the moving frame are computed using Rodrigues's formula for exponentiation of a skew-symmetric matrix. Cumulative sums of the resulting tangents provide the final numerical approximations to $\Gamma[\phi]$.

\subsection{Properties of $\Gamma[\Phi[x]]$}

We now discuss how the properties of the map $\Phi$ translate to properties of the map $\Gamma\circ \Phi$. By construction, the tangent of $\Gamma[\phi]$ is normalized, so the curve is parameterized by arc-length. Moreover, the equations for the Bishop frame specify
\begin{align}
\frac{d}{dt} {\bf T}(t) = \phi_1(t) {\bf N}_1(t) + \phi_2(t) {\bf N}_2(t)
\end{align}
where ${\bf N}_i$ are always orthonormal. Consequently, the curvature function is
\begin{align}
\kappa_{\Gamma[\phi]}(t) = \left\Vert \frac{d}{dt} {\bf T}(t)\right\Vert = \sqrt{\phi_1^2(t)+\phi_2(t)}
\end{align}
for all $t$. 

Now, if $\Phi$ is an isotropic isometry, we have that the total square curvature is 

\begin{align}
\int_0^1 \kappa_{\Gamma[\Phi[x]]}(t)^2\:dt = \Vert \Phi[x]\Vert_{L^2}^2 = 2\Vert x\Vert^2. 
\end{align}
Moreover, the isotropic isometry implies that
\begin{align}
\int_0^1 \kappa_{\Gamma[\Phi[x]]}(t)^2\:dt = \Vert \Phi[x]\Vert_{L^2}^2 = 2\Vert x\Vert^2. 
\end{align}
Therefore, for any fixed $x_0\in\R^d$, we have that 
\begin{align}
\int_0^1 \kappa_{\Gamma[\Phi[x-x_0]]}(t)^2\:dt = 2\Vert x-x_0\Vert^2. 
\end{align}

Consequently, we can use the isotropic isometry property to visually assess divergence from a base point by observing the curvature. Unfortunately, it is much harder to obtain metric comparison results between individual curves after the non-linear transformation based on the isometric isometry property. 

The asymptotic tour property heuristically ensures that the tangents of filaments may exhibit diversity of directionality. While a full analysis of such behavior is beyond the technical scope of this paper, we consider some informal heuristics. For small $\Delta t$, we have that
\begin{align}
U^\prime(\Delta t) \approx A(0) U(\Delta t),
\end{align}
and hence
\begin{align}
U(\Delta t) \approx \text{exp}(\Delta t A(0)) U(0) = \text{exp}(\Delta t A(0)) 
\end{align}
where $\text{exp}$ denotes the matrix exponential and 
\begin{align}
A(0) = \begin{pmatrix}
0 & \phi_1(0) & \phi_2(0)\\
-\phi_1(0) & 0 & 0\\
-\phi_2(0) & 0 & 0
\end{pmatrix}.
\end{align}
The Rodrigues formula gives us

\begin{align}
\text{exp}(\Delta t A(0))=I + \sin(\rho)\begin{pmatrix}
0 & c & s\\
-c & 0 & 0\\
-s & 0 & 0
\end{pmatrix} + (1-\cos(\rho)) \begin{pmatrix}
-1 & 0 & 0\\
0 & -c^2& -cs\\
0 & -cs & -s^2
\end{pmatrix}
\end{align} 
where $\rho = \Delta t \sqrt{\phi_1(0)^2+\phi_2(t)^2}$, $c= \phi_1(0)/\rho$, and $s=\phi_2(0)/\rho$. This gives a perturbative expression for the tangent:
\begin{align}
\begin{pmatrix}
\cos(\rho) &
c\sin(\rho) &
s\sin(\rho)
\end{pmatrix}
\end{align}
The tour property ensures that the perturbative directions of the form 
\begin{align}
\begin{pmatrix}
0 &
c&
s
\end{pmatrix}
\end{align}
 can exhibit independence, leading to a possible diversity of motions in the underlying curves.

Now, what does the minimization of the mean quadratic variation entail in the case of curves for filament plots? Formally, given $\phi_1(t)$ and $\phi_2(t)$, the Frenet-Serret frame of the system has normal vector
\begin{align}
{\bf N}(t) = \frac{\phi_1(t)}{\sqrt{\phi_1(t)^2+\phi_2(t)^2}}{\bf N}_1(t) + \frac{\phi_2(t)}{\sqrt{\phi_1(t)^2+\phi_2(t)^2}}{\bf N}_2(t),
\end{align} 
and binormal vector
\begin{align}
{\bf B}(t) = -\frac{\phi_2(t)}{\sqrt{\phi_1(t)^2+\phi_2(t)^2}}{\bf N}_1(t) + \frac{\phi_1(t)}{\sqrt{\phi_1(t)^2+\phi_2(t)^2}}{\bf N}_2(t).
\end{align} 
When we take a formal derivative of ${\bf B}$, we have 
\begin{align}
\frac{d}{dt}{\bf B}(t) =& -\frac{\phi_2^\prime (t)\left(\phi_1(t)^2+\phi_2(t)^2\right) - \phi_2(t)\left(\phi_1^\prime \phi_1 + \phi_2^\prime\phi_2\right)}{[\phi_1(t)^2+\phi_2(t)^2]^{3/2}}{\bf N}_1(t)\\
&+\frac{\phi_1^\prime (t)\left(\phi_1(t)^2+\phi_2(t)^2\right) - \phi_1(t)\left(\phi_1^\prime \phi_1 + \phi_2^\prime\phi_2\right)}{[\phi_1(t)^2+\phi_2(t)^2]^{3/2}}{\bf N}_2(t)\\
& + \frac{\phi_2(t)\phi_1(t)}{\sqrt{\phi_1(t)^2+\phi_2(t)^2}}{\bf T}(t) - \frac{\phi_1(t)\phi_2(t)}{\sqrt{\phi_1(t)^2+\phi_2(t)^2}}{\bf T}(t).\\
=& - \frac{\phi_2^\prime (t)\phi_1(t) - \phi_2(t)\phi_1^\prime(t)}{\phi_1(t)^2+\phi_2(t)^2} {\bf N}(t)
\end{align}
By definition, the torsion function for the curve $\Gamma[\phi]$ is then
\begin{align}
\tau_{\Gamma[\phi]}(t) = \frac{\phi_2^\prime (t)\phi_1(t) - \phi_2(t)\phi_1^\prime(t)}{\phi_1(t)^2+\phi_2(t)^2}
\end{align}
whenever this quantity is defined. One may then formally verify the identity
\begin{align}
\phi_1^\prime(t)^2 + \phi_2^\prime(t)^2 = \kappa_{\Gamma[\phi]}^\prime(t)^2 + \tau_{\Gamma[\phi]}(t)^2 \kappa_{\Gamma[\phi]}(t)^2
\end{align}

Therefore, minimization of the mean quadratic variation of $\{\Phi[x_i]\}_{i=1}^N$ is formally equivalent to minimizing the sum of the mean quadratic variation of the curvature function and the $L^2$ inner product of the squared torsion function with the squared curvature function. Small quadratic variation of curvature means that the curvature should not change very much. On the other hand, a small value for $\Vert  \tau_{\Gamma[\phi]}^2 \kappa_{\Gamma[\phi]}^2\Vert_{L^2([0,1])}$ precludes sustained periods of large curvature and torsion. Visually, this means that the curve is nearly planar during sustained periods of large curvature.  

To get further intuition on the minimization of this quantity, consider the extreme case where the integral of the square is zero. Then $\kappa_{\Gamma[\phi]}^\prime$ is zero, so $\kappa_{\Gamma[\phi]}^\prime$ is constant. If $\kappa_{\Gamma[\phi]}^\prime$ is zero, then the curve is a line. Otherwise, the torsion $\tau_{\Gamma[\phi]}(t)$ is zero and the curve is a planar circle.

\subsection{Boston dataset visualizations}
The Boston Housing dataset \cite{harrison1978hedonic} provides a 13-dimensional example with a single target for regression (median house price in a given statistical area). Figures \ref{fig:bostonAndrews} and \ref{fig:bostonFilament} consider labeling this dataset according to the decile of the median house price over the dataset. The 3D plots prove difficult to read, but the filament plot visually separates individual points while retaining the expected cluster structure.

\begin{figure}[ht]
\centering
\includegraphics[scale=0.5]{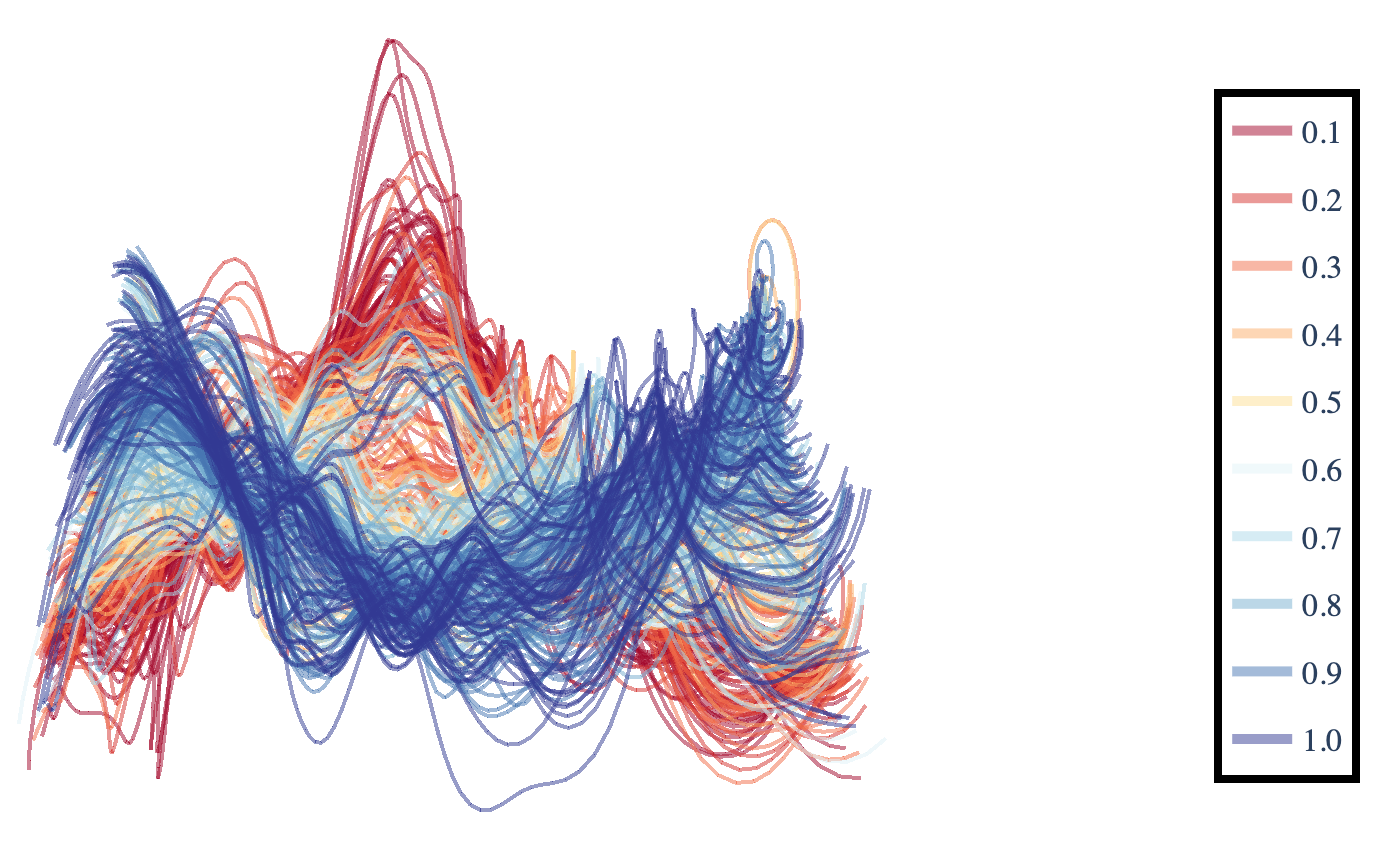}
\caption{The graphs of the 3D Andrews plots for the Boston dataset. Deciles of the median house price provide are used to color the different filaments.}
\label{fig:bostonAndrews}
\end{figure}

\begin{figure}[ht]
\centering
\includegraphics[scale=0.5]{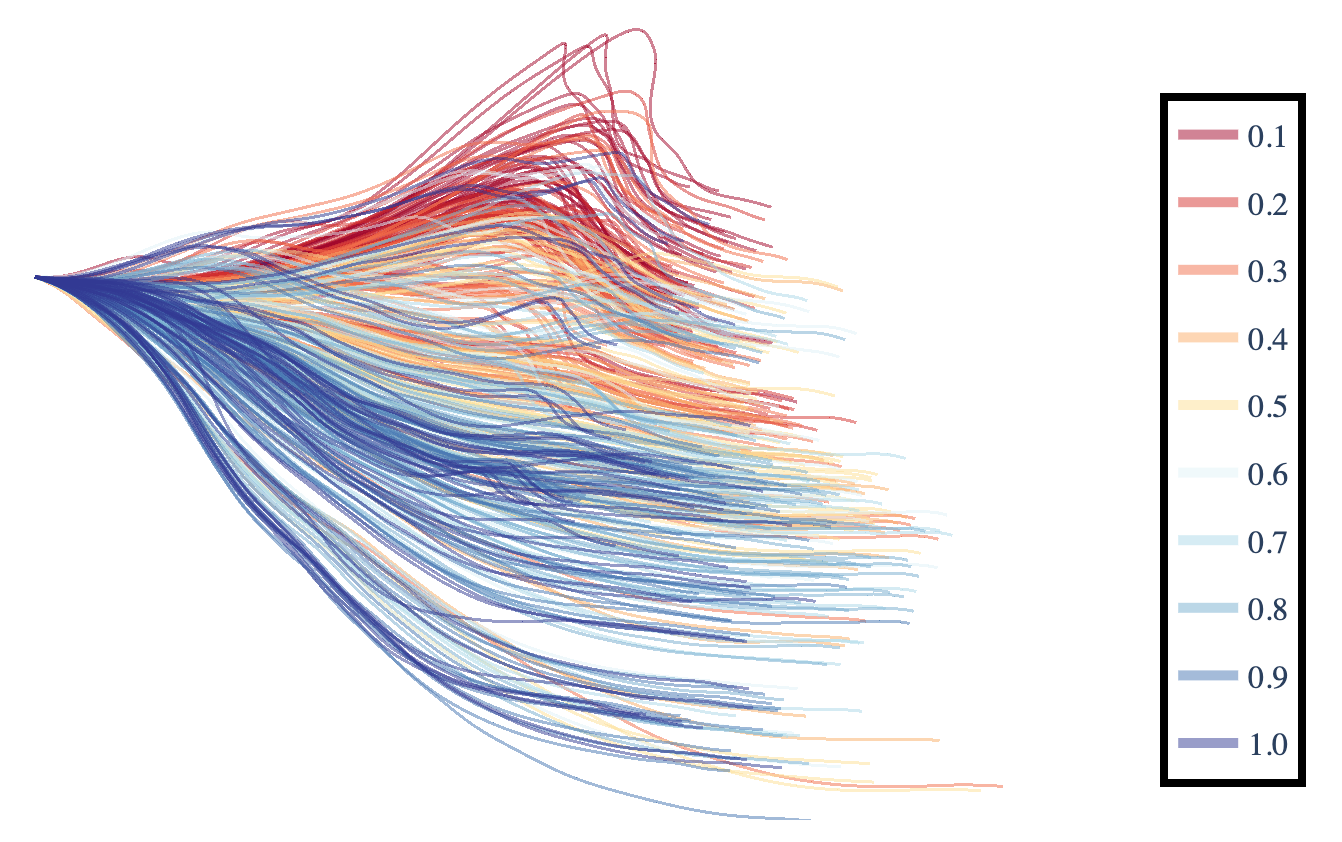}
\caption{The Filament plot for the Boston housing dataset. Note that this is just a single perspective, so interested readers should view the code to see the plot in a ``drag-to-rotate" interface.}
\label{fig:bostonFilament}
\end{figure}

\subsection{Breast cancer dataset visualizations}
Next, we consider the version of the Wisconsin breast cancer dataset \cite{street1993nuclear} provided in the python sklearn package.  This dataset has 30 dimensions and two classes. In contrast to the 3D Andrews plots in Figure \ref{fig:bcAndrews}, the filament plot in Figure \ref{fig:bcFilament} illustrates that the classes are mostly separable with some notable exceptions.

\begin{figure}[ht]
\centering
\includegraphics[scale=0.5]{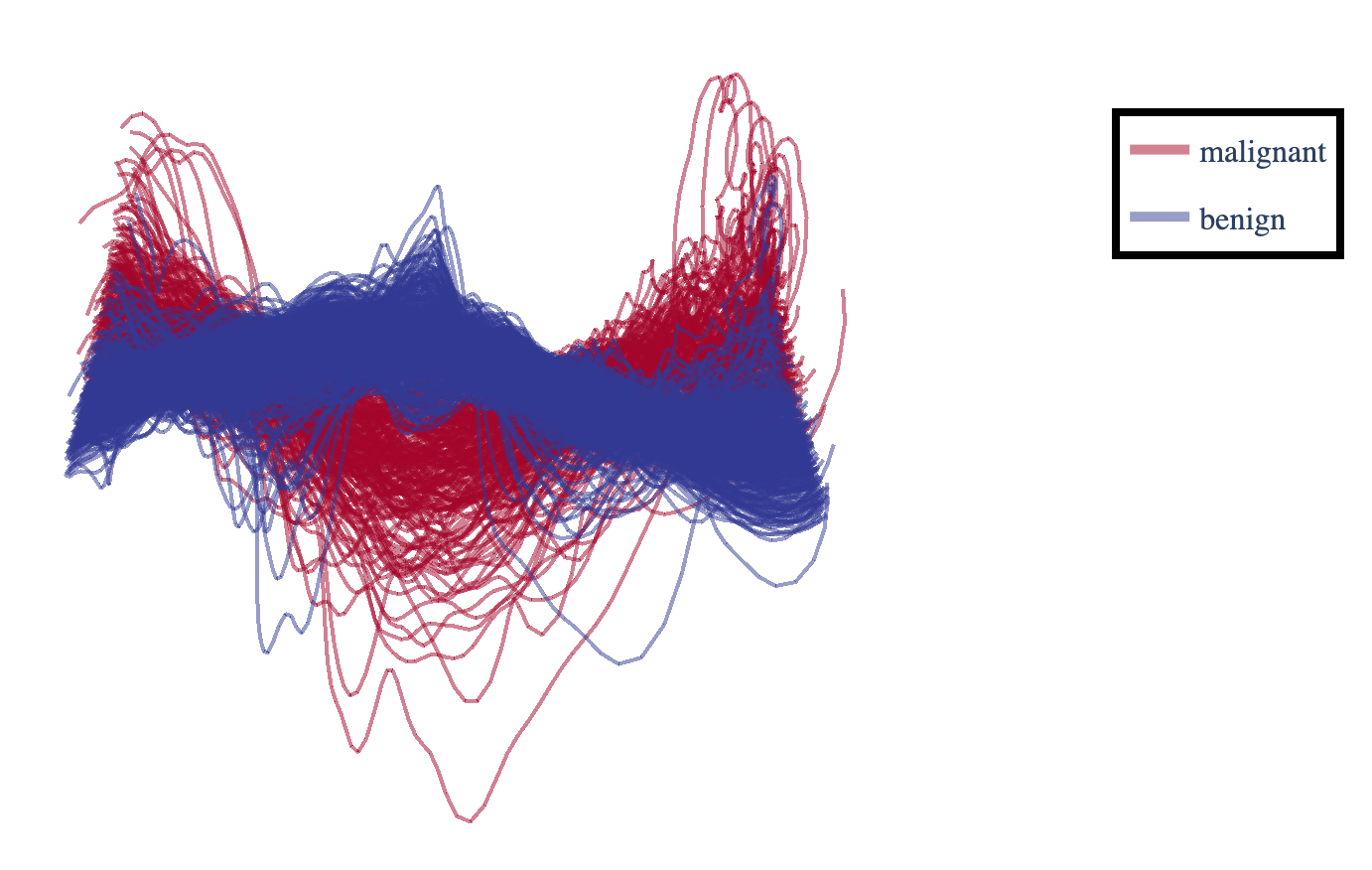}
\caption{The graphs of the 3D Andrews plots for the Wisconsin breast cancer dataset.}
\label{fig:bcAndrews}
\end{figure}

\begin{figure}[ht]
\centering
\includegraphics[scale=0.5]{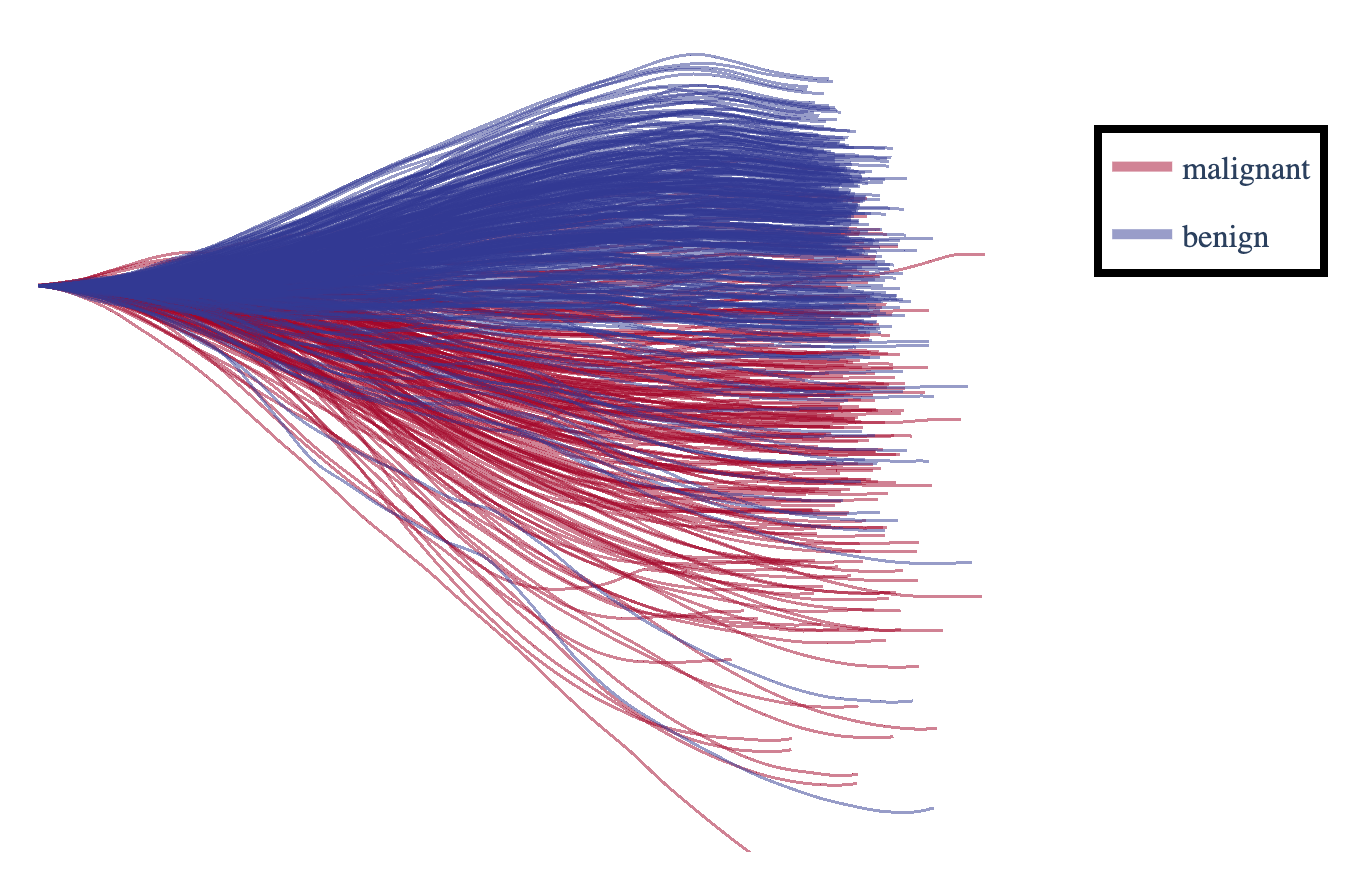}
\caption{The Filament plot for the Wisconsin breast cancer dataset. Note that this is just a single perspective, so interested readers should view the code to see the plot in a ``drag-to-rotate" interface.}
\label{fig:bcFilament}
\end{figure}

\subsection{Digits dataset visualizations}
Our last example is 64 dimensional space \cite{xu1992methods} consisting of $8$ by $8$ pixel images of handwritten digits with 10 total classes and 1797 examples. While the logic of visualizing image datasets seems tortured, images only admit ``small-multiple" visualizations, so scatterplots provide the ability to compare proximity across a larger slice of the dataset. The 3D Andrews plots in Figure \ref{fig:digitsAndrews} are quite difficult to read, but the filament plot in Figure \ref{fig:digitsFilament} separates the data nicely while indicating the proximity of similar classes (for example, $4$, $7$, and $9$ have interesting overlapping structure).

\begin{figure}[ht]
\centering
\includegraphics[scale=0.5]{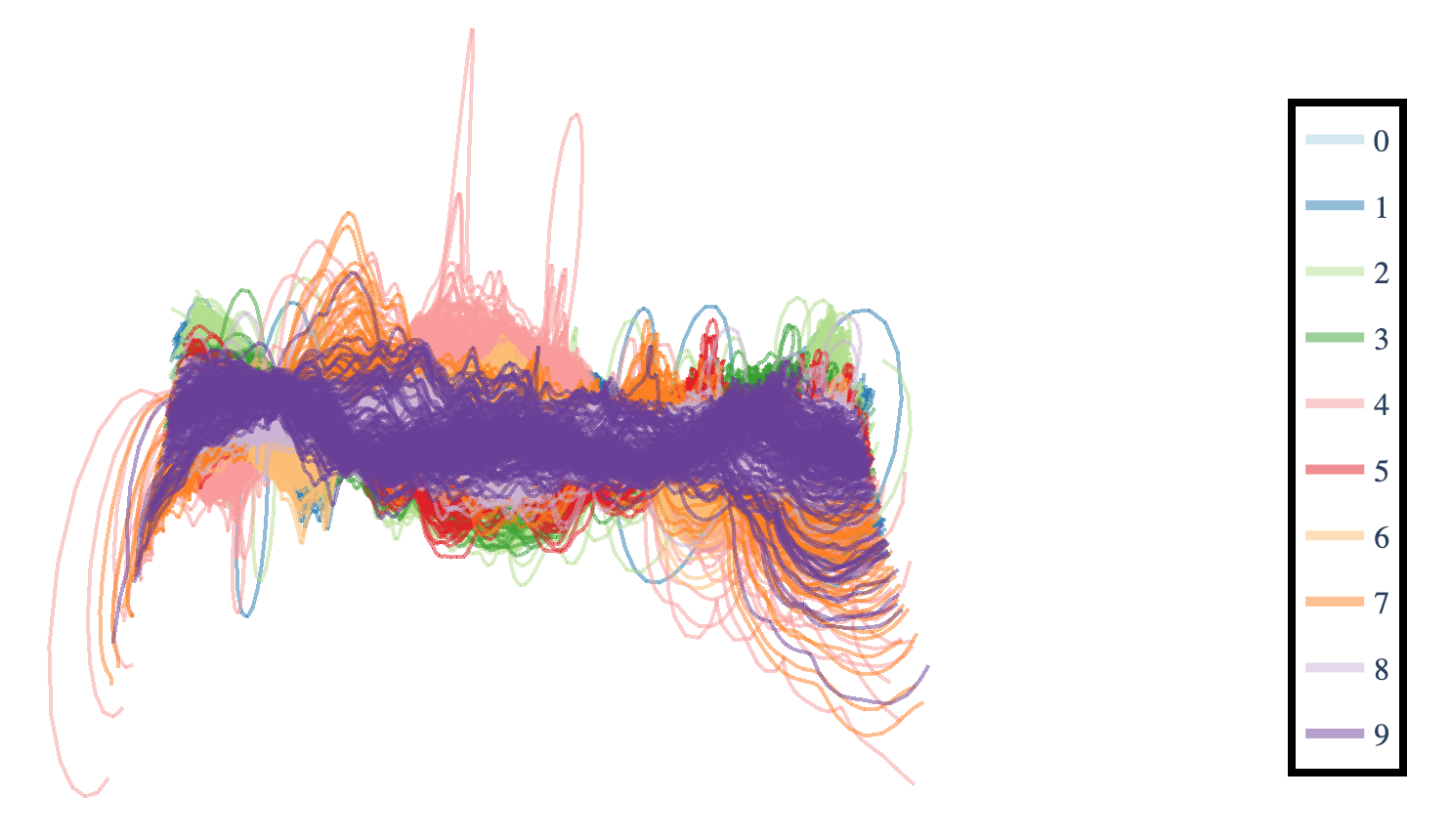}
\caption{The graphs of the 3D Andrews plots for the digits dataset.}
\label{fig:digitsAndrews}
\end{figure}

\begin{figure}[ht]
\centering
\includegraphics[scale=0.5]{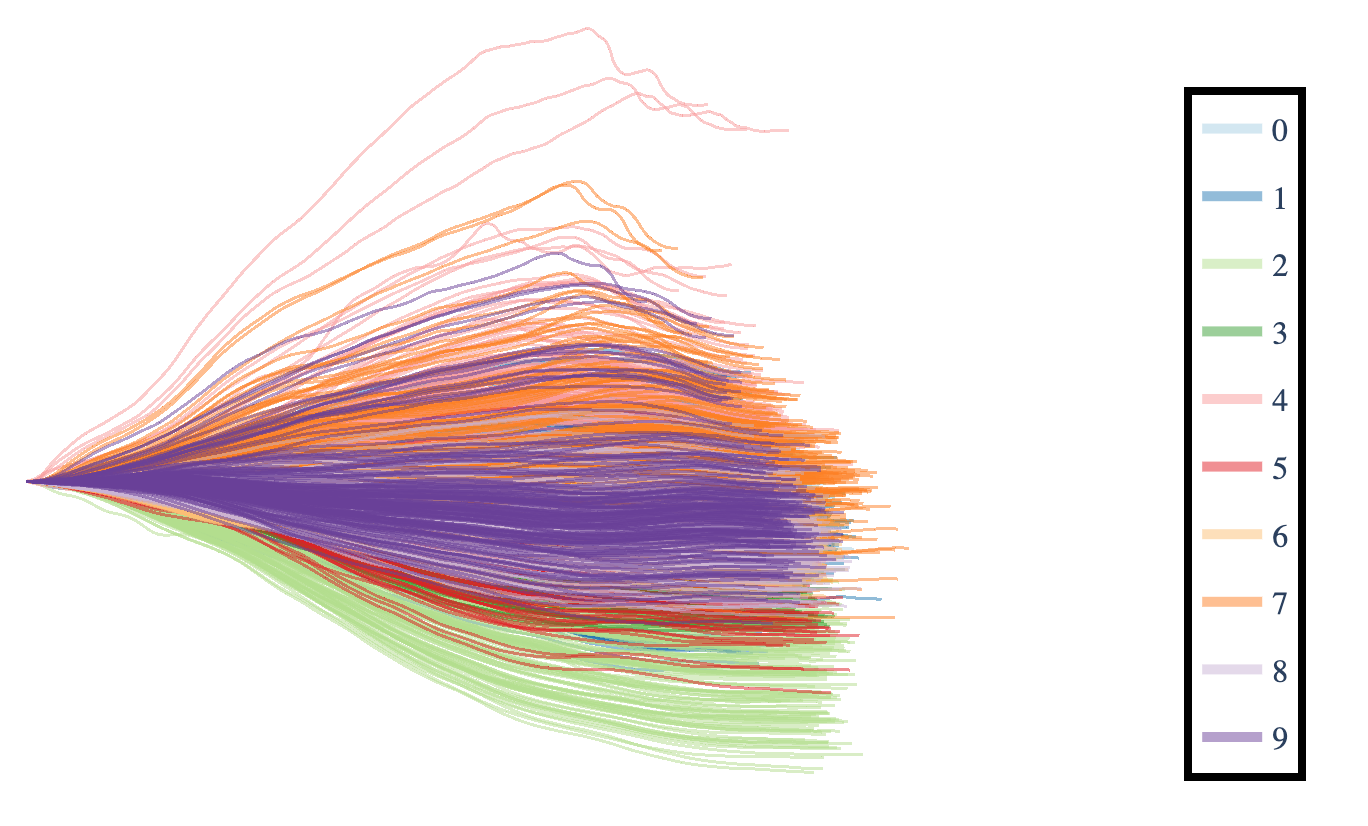}
\caption{The Filament plot for the digits dataset. Note that this is just a single perspective, so interested readers should view the code to see the plot in a ``drag-to-rotate" interface. }
\label{fig:digitsFilament}
\end{figure}

\section{Conclusions}
\label{sec:conclusion}

There are several additional question that we do not address in this work, but which are interesting topics for future exploration.\\

\paragraph{$L^1$ and $L^\infty$ gap} While the $L^1$ and $L^\infty$ bounds provide useful visual heuristics, the gap between these bounds grows like $\sqrt{d}$ for our embeddings. This follows from various resolutions of the Littlewood conjecture \cite{trigub2003lower} in the $L^1$ case and Erdos \cite{erdos1962inequality} indicates the sharpness of the result in the $L^\infty$ case. We leave it as an open problem to construct embeddings which shrink this gap. Additionally, we leave it as an open problem to construct data embeddings into curves that result in comparison bounds in the Hausdorff metric.\\

\paragraph{Robust visualizations} The form of solutions in Theorem \ref{thm:1} relies on the SVD (which is equivalent to PCA after mean shifting), and is therefore subject to robustness issues. While variants of the Davis-Kahan Theorem \cite{davis1970rotation} indicate some measure of stability for PCA eigenspaces, in general the instability of the eigenvectors translates to instability of frequencies for these 3D Andrews plots. Therefore, any visual assessment based on frequency information may be suspect. While we do not directly address these issues, we shall consider ways to mitigate them in future work, and it should also be noted that the isometry property ensures that visual interpretations in terms of the $L^1$-$L^\infty$ bounds are robust.\\

\paragraph{Embeddings of metric spaces into $L^\infty$} The $\sqrt{d}$ gap mentioned above and the result of \cite{matouvsek1990bi} on $O(N^{2/d}\log^{3/2}(N))$-distortion embeddings of $N$-point metric spaces into $\R^d$ suggests that $N$ point metric spaces may be embedded into spaces of functions with a $L^1$-$L^\infty$ gap on the order of $O(\log^{5/4}(N))$. \cite{linial1995geometry} indicates that certain expander graphs cannot embed without $\Omega(\log(N))$ distortion for $\ell_p^d$ with $1\leq p\leq 2$ no matter the dimension of the target space. However, the same paper indicates in Lemma 3.1 that we can embed any finite metric space on $N$ points into $\ell_\infty^N$. This suggests that we may embed into spaces of curves under the Hausdorff or $L^\infty$ distance with no distortion, but possibly at the cost of smoothness. 

On the other hand, work such as \cite{badoiu2005low, matouvsek2010inapproximability} indicates that many of these low-distortion embeddings are difficult to compute. The method we present only requires an SVD and simple numerical integration.\\

\paragraph{Closed curve problem} Because the linear component of a space curve exhibits the most visual impact, the removal of this component generally reduces biases that would appear in visualizations. Thus, it is interesting to consider embeddings into spaces of closed curves. This may be done by simply taking the image of a periodic curve, but if our goal is to embed into curvatures that induce closed curves, we must confront the \emph{closed curve problem}. 

The closed curve problem seeks necessary and sufficient conditions on the curvature (and torsion) functions for closed curves. While \cite{arroyo2008periodic} resolves the problem for plane curves in many cases, little progress has been made on closed space curves due to the substantial increase in complexity.

\paragraph{Machine learning applications} This paper only consider data visualization applications, but filament plots may provide a useful transformation that may be coupled with functional regression/classification techniques, or even convolutional neural networks to perform regression or classification. Such an exploration lies outside the scope of this paper, but we look forward to investigating the utility of such procedures in the future.

\section*{Acknowledgements}
We would like to thank Radu Balan for discussions involving solutions to the Bishop frame equations. We would also like to thank anonymous reviewers for several suggestions that greatly improved the quality of this paper and the clarity of the proofs.

\vskip 0.2in
\bibliographystyle{plain}
\bibliography{filaments}

\end{document}